%% file: main.tex
\documentclass[a4paper,UKenglish,cleveref, autoref, thm-restate]{lipics-v2021}

\pdfoutput=1 
\hideLIPIcs  

\graphicspath{{./images/}}

\title{Revisiting Token Sliding on Chordal Graphs} 

\author{Rajat Adak}{Indian Institute of Science Bangalore, India  \and \url{https://sites.google.com/view/rajatadak}}{rajatadak@iisc.ac.in}{}{}

\author{Saraswati Girish Nanoti}{Indian Institute of Technology Gandhinagar, India}{nanoti_saraswati@iitgn.ac.in}{}{}

\author{Prafullkumar Tale}{Indian Institute of Science Education and Research Bhopal, India \and \url{https://pptale.github.io/}}{prafullkumar@iiserb.ac.in}{}{}

\authorrunning{Adak, Nanoti and Tale} 

\Copyright{Rajat Adak, Saraswati Girish Nanoti and Prafullkumar Tale} 



\category{} 

\relatedversion{} 

\funding{}

\acknowledgements{}

\usepackage[svgnames]{xcolor}

\usepackage{etoolbox}
\makeatletter
\patchcmd{\BR@backref}{\newblock}{\newblock($\uparrow$~}{}{}
\patchcmd{\BR@backref}{\par}{)\par}{}{}
\makeatother

\PassOptionsToPackage{numbers,sort&compress}{natbib}
\usepackage{natbib}
\hypersetup{
    colorlinks=true,
    linkcolor=IndianRed,
    filecolor=magenta,      
    urlcolor=DodgerBlue,
    citecolor = SeaGreen
}

\EventEditors{John Q. Open and Joan R. Access}
\EventNoEds{2}
\EventLongTitle{42nd Conference on Very Important Topics (CVIT 2016)}
\EventShortTitle{CVIT 2016}
\EventAcronym{CVIT}
\EventYear{2016}
\EventDate{December 24--27, 2016}
\EventLocation{Little Whinging, United Kingdom}
\EventLogo{}
\SeriesVolume{42}
\ArticleNo{23}

\input{preamble}

\begin{document}

\maketitle

\begin{abstract}
In this article, we revisit the complexity of the reconfiguration 
of independent sets under the token sliding rule on chordal graphs.
In the \textsc{Token Sliding-Connectivity} problem, 
the input is a graph $G$ and an integer $k$, and 
the objective is to determine whether the reconfiguration graph 
$TS_k(G)$ of $G$ is connected. 
The vertices of $TS_k(G)$ are $k$-independent sets of $G$, 
and two vertices are adjacent if and only if one can transform one 
of the two corresponding independent sets into the other by sliding 
a vertex (also called a \emph{token}) along an edge. 
Bonamy and Bousquet [WG'17] proved that the 
\textsc{Token Sliding-Connectivity} problem is 
polynomial-time solvable on interval graphs but \NP-hard 
on split graphs.
In light of these two results, the authors asked: can we decide the connectivity of $TS_k(G)$ in polynomial time for chordal graphs with 
\emph{maximum clique-tree degree} $d$? 
We answer this question in the negative and prove that the problem is \para-\NP-hard when parameterized by $d$. 
More precisely, the problem is \NP-hard even when $d = 4$. 
We then study the parameterized complexity of the problem for a larger 
parameter called \emph{leafage} and prove that the problem is 
\co-\W[1]-hard.
We prove similar results for a closely related problem called 
\textsc{Token Sliding-Reachability}. 
In this problem, the input is a graph $G$ with two of its 
$k$-independent sets $I$ and $J$, and 
the objective is to determine whether there is a sequence of 
valid token sliding moves that transform $I$ into $J$.
\end{abstract}

\input{introduction}

\input{preliminaries.tex}
\section{Parameterized by Maximum Clique-Tree Degree}
\input{reduction-bounded-clique-deg}

\section{Parameterized by leafage}
\label{sec:leafage-hardness}
\input{TS-conn-W1hard.tex}
\input{TS-reach-W1hard.tex}
\input{conclusion}

\bibliography{references}

\end{document}

%% file: preamble.tex
\usepackage{amssymb}
\usepackage{amsmath}
\usepackage{amsthm}
\usepackage{hyperref}
\usepackage{xspace}
\usepackage{mathtools}
\nolinenumbers
\usepackage[algoruled,boxed,lined]{algorithm2e}
\usepackage{complexity} 
\usepackage{todonotes} 
\presetkeys{todonotes}{inline}{}
\usepackage{comment}
\usepackage[capitalise,noabbrev]{cleveref}
\usepackage{tikz}
\usetikzlibrary{shapes.geometric}
\usetikzlibrary{positioning}

\usepackage{color}

\newcommand{\calC}{\mathcal{C}}

\newcommand{\calH}{{\mathcal H}}

\newcommand{\calM}{\ensuremath{{\mathcal M}}}
\newcommand{\calO}{\ensuremath{{\mathcal O}}}
\newcommand{\OO}{\mathcal{O}}

\newcommand{\calT}{\mathcal{T}}

\newcommand{\calV}{\mathcal{V}}

\newcommand{\para}{\textsf{para}}

\newcommand{\NPH}{\NP-hard\xspace}
\newcommand{\WoneH}{\textup{\textsf{W[1]}}-hard\xspace}

\newcommand{\hard}{{hard}}

\newcommand{\yes}{\textsc{Yes}}
\newcommand{\no}{\textsc{No}}

\newtheorem*{proposition*}{Proposition}
\newtheorem{reduction rule}[theorem]{Reduction Rule}
\Crefname{reduction rule}{Reduction~Rule}{Reduction~Rules}
\newtheorem{greedy select}[theorem]{Greedy~Select}\Crefname{greedy select}{Greedy~Select}{Greedy~Select}

\newtheorem{marking-scheme}[theorem]{Marking Scheme}

\newcommand{\defproblem}[3]{
  \vspace{1mm}
\noindent\fbox{
  \begin{minipage}{0.96\textwidth}
  \begin{tabular*}{\textwidth}{@{\extracolsep{\fill}}lr} #1 \\ \end{tabular*}
  {\bf{Input:}} #2  \\
  {\bf{Question:}} #3
  \end{minipage}
  }
  \vspace{1mm}
}



%% file: introduction.tex
\section{Introduction}

A typical reconfiguration problem on a graph consists of two feasible 
solutions $S$, $T$, and a modification rule. 
The objective is to determine whether $S$ can be transformed 
{into} $T$ by repeated applications of the modification rule 
such that at any intermediate step, the set remains a 
feasible solution. 
It {is helpful to imagine} that a token is located {at} every vertex of $S$. 
{Therefore,} the modification operation corresponds to how 
the tokens can be moved. 
Three well-studied operations are 
\emph{token addition and removal}, \emph{token jumping}, and 
\emph{token sliding}. 
In the first operation, one can add or remove the token 
{during} the intermediate stages. 
In the second operation, a token can jump from one vertex to another vertex (usually to a vertex that doesn't have a token). 
In \emph{token sliding}, one can move tokens along edges of the graph, i.e., vertices can only be replaced with {adjacent} vertices.

The reconfiguration version of various classical problems like 
\textsc{Satisfiability}~\cite{DBLP:journals/siamcomp/GopalanKMP09}, 
\textsc{Dominating Set}, \textsc{Independent Set}, 
\textsc{Colouring's}~\cite{DBLP:journals/endm/BonamyB13,DBLP:journals/jgt/FeghaliJP16}, 
\textsc{Kempe Chains}~\cite{DBLP:journals/jct/BonamyBFJ19,DBLP:journals/ejc/Feghali0P17}, 
\textsc{Shortest Paths}~\cite{DBLP:journals/tcs/Bonsma13}, etc., 
{has received significant attention in recent years.} 
{See surveys~\cite{DBLP:journals/algorithms/Nishimura18,DBLP:books/cu/p/Heuvel13} for a comprehensive introduction to these topics.} 
Among these problems, \textsc{Independent Set Reconfiguration} is 
{one of} {the most well-studied} problems with respect to 
all the operations mentioned above, 
{namely} token addition and removal~\cite{DBLP:journals/tcs/KaminskiMM12,DBLP:journals/algorithmica/MouawadN0SS17}, 
token jumping~\cite{DBLP:conf/swat/BonsmaKW14,DBLP:conf/fct/BousquetMP17,DBLP:journals/tcs/ItoDHPSUU11,DBLP:conf/tamc/ItoKOSUY14,DBLP:conf/isaac/ItoKO14}, and 
token sliding~\cite{DBLP:conf/wg/BonamyB17,DBLP:journals/tcs/DemaineDFHIOOUY15,DBLP:conf/isaac/Fox-EpsteinHOU15,DBLP:journals/tcs/HearnD05,DBLP:conf/isaac/HoangU16,DBLP:journals/talg/LokshtanovM19}. 
Recall that for an integer $k$, a $k$-independent set of $G$ is a subset 
$S \subseteq V(G)$ of size $k$ of pairwise non-adjacent vertices. 
In this article, we consider the complexity of the reconfiguration of 
\textsc{Independent Set} problem under the token sliding rule. 
As we focus only on {\textsc{Independent Set},} we refer to it simply as \textsc{Token Sliding} and consider the following two problems.

\defproblem{\textsc{Token Sliding-Connectivity (TS-Connectivity)}}{Graph $G$ and integer $k$}{
Is the $k$-$TS$ reconfiguration graph $TS_k(G)$ of $G$ connected?

The vertices of $TS_k(G)$ are $k$-independent sets of $G$ and 
$(I, J)$ is an edge of $TS_k(G)$ if and only if 
$J \setminus I = \{u\}$, $I \setminus J = \{v\}$ and 
$(u, v)$ is an edge of $G$ for some $u, v\in V(G)$.}

\defproblem{\textsc{Token Sliding-Rechability (TS-Rechability)}}{Graph $G$ and
two $k$-independent set of $I, J$ of $G$.}{Are the vertices corresponding to $I$ and $J$ in the same connected components of $TS_k(G)$?
}

\noindent 
Alternately, in the second variation, the objective is to determine whether
we can transform one independent set to the other by 
token sliding, while the set remains independent.

Hearn and Demaine~\cite
{DBLP:journals/tcs/HearnD05} proved  
that \textsc{TS-Reachability}
is \PSPACE-complete even for planar graphs. 
On the positive side, Kaminski et al.~\cite{DBLP:journals/tcs/KaminskiMM12} and Bonsma et al.~\cite{DBLP:journals/mst/BelmonteKLMOS21} presented polynomial time algorithms 
when the input graph is restricted to cographs and claw-free graphs,
respectively.
Yamada and Uehara~\cite{DBLP:conf/walcom/YamadaU16} showed that any 
two $k$-independent sets
in proper interval graphs can be transformed to each other 
in polynomial steps.
Demaine et al.~\cite{DBLP:conf/isaac/DemaineDFHIOOUY14} described a linear algorithm when
the input graph is a tree.
In the same paper, the authors asked if determining whether 
\textsc{TS-Reachability} is polynomial time for 
interval graphs and then more generally for chordal graphs, 
a strict superclass of interval graphs. 
Belmonte et al.~\cite{DBLP:journals/mst/BelmonteKLMOS21} answered the second
part in the negative by proving that the problem is 
\PSPACE-complete on split graphs, a strict subset
of chordal graphs.
Bonamy and Bousquet~\cite{DBLP:conf/wg/BonamyB17} answered 
the first part 
positively by presenting a polynomial-time 
algorithm for \textsc{TS-Reachability} when input is an
interval graph.
They also proved the following result about  
\textsc{TS-Connectivity}.

\begin{proposition*}[Theorem~$2$ and~$3$ in \cite{DBLP:conf/wg/BonamyB17}]
\label{prop:token-sliding-known-results}
\emph{\textsc{Token Sliding-Connectivity}}
is in \emph{\P} for interval graphs, but
is \emph{\co-\NP-hard}
for split graphs.
\end{proposition*}

It is known that for each chordal graph there is a tree, 
called \emph{clique-tree}, whose
vertices are the cliques of the graph and for every vertex 
in the graph, the set of cliques that contain it form a subtree of clique-tree.
See, for example, the note by Blair and Peyton~\cite{10.1007/978-1-4613-8369-7_1}.
The maximal cliques in an interval graph can be ordered such 
that for every vertex, the maximal cliques containing that vertex occur consecutively~\cite{DBLP:journals/algorithmica/GalbyMSST24}.
Hence, for interval graphs, such the maximum clique-tree degree is two.
Whereas, for split graphs the maximum degree of a clique-tree 
can be unbounded, as demonstrated by a star graph.
This prompted Bonamy and Bousquet~\cite{DBLP:conf/wg/BonamyB17}
to explicitly ask the following question.

\smallskip
\noindent \textbf{Question:} \emph{Let $d$ be a fixed constant. 
For an integer $k$ and chordal graph $G$ of the maximum
clique-tree degree at most $d$, 
can the connectivity of $TS_k(G)$ be decided in polynomial time?}
\smallskip

The questions like this have been systematically addressed in the realm 
of \emph{parameterized complexity}.
This paradigm allows for a more refined analysis of the problem’s complexity.
In this setting, we associate each instance $I$ with a parameter $\ell$
which can be arbitrarily smaller than the input size.
The objective is to check whether the problem is `tractable' when 
the parameter is `small'.
A parameter may originate from the formulation of the problem itself (called \emph{natural parameters}) or it can be a property of the input graph (called \emph{structural parameters}).
Parameterized problems that admit an algorithm with running time
$f(\ell) \cdot |I|^{\OO(1)}$ for some computable function $f$, are called 
\emph{fixed parameter tractable} (\FPT).
On the other hand, parameterized problems that are hard for the complexity class \W[1] do not admit such fixed-parameter algorithms,
under standard complexity assumptions.
A problem is \para-\NP-hard if it remains \NP-hard even when the parameter 
$\ell$ is fixed. 
In other words, fixing the parameter does not simplify the problem to a
polynomial-time solution.

For the above problems, the natural parameter is the number of tokens $k$
whereas the maximum clique-tree degree 
is one of the structural parameters. 
It is known that \textsc{TS-Connectivity}~\cite{DBLP:conf/wg/BonamyB17}
and \textsc{TS-Reachability}~\cite{DBLP:conf/tamc/ItoKOSUY14} are \co-\W[2]-hard and
\W[1]-hard, respectively, when parameterized by the number of tokens $k$.
However, the latter problem admits an \FPT~algorithm for planar graphs~\cite{DBLP:conf/isaac/ItoKO14}.
In this article, we study the structural parameterization
of these problems when the input is a chordal graph.
The first parameter we consider is the maximum clique-tree
degree of the chordal graph and answer the question by 
Bonamy and Bousquet~\cite{DBLP:conf/wg/BonamyB17}
in the negative.

\begin{restatable}{theorem}{tsconnnphard}
\label{thm:np-hardness-clique-tree-degree-TS-Conn}
On chordal graphs of maximum clique-tree degree $4$,
\emph{\textsc{Token Sliding-Connectivity}} problem is 
\emph{\co-\NP-hard}.
\end{restatable} 

Bonamy and Bousquet~\cite{DBLP:conf/wg/BonamyB17} 
adapted their algorithm for \textsc{TS-Connectivity} on interval graphs
to solve \textsc{TS-Reachability} on interval graphs in polynomial time.
Also, Belmonte et al.~\cite{DBLP:journals/mst/BelmonteKLMOS21} 
proved that \textsc{TS-Reachability} is \PSPACE-complete on split graphs.
Hence, it is natural to ask the analogous question regarding
\textsc{TS-Reachability}: 
\emph{Does \emph{\textsc{TS-Reachability}} admit a 
polynomial-time algorithm when the input is a chordal graph 
of the maximum clique-tree degree at most $d$?}
We answer even this question in the negative.

\begin{restatable}{theorem}{tsreachnphard}
\label{thm:np-hardness-clique-tree-degree-TS-Reach}
On chordal graphs of maximum clique-tree degree $3$,
\emph{\textsc{Token Sliding-Reachability}} problem is 
\emph{\NP-\hard}.
\end{restatable}

Next, we move to our second structural parameter \emph{leafage}.
The {leafage} of chordal graph $G$, denoted by $\ell$, is defined as 
the minimum number of leaves in the tree of a tree representation of $G$.
Habib and Stacho~\cite{DBLP:conf/esa/HabibS09} showed that we can 
compute the leafage of a connected chordal graph in polynomial time. 
In recent years, researchers have studied the structural parameterization of various graph problems on chordal graphs parameterized by the leafage.
Fomin et al.~\cite{DBLP:journals/algorithmica/FominGR20} and Arvind et al.~\cite{DBLP:journals/corr/ArvindNPZ22} proved, respectively, that the \textsc{Dominating Set} and \textsc{Graph Isomorphism} problems on chordal graphs are \FPT\ parameterized by the leafage.
Galby et al.~\cite{DBLP:journals/algorithmica/GalbyMSST24} presented 
an improved \FPT\ algorithm for {\sc Dominating Set} parameterized by the 
leafage, and adapted it to obtain 
similar \FPT\ algorithms for {\sc Connected Dominating Set} and 
{\sc Steiner Tree} on chordal graphs.
They also proved that
{\sc MultiCut with Undeletable Terminals} on chordal graphs
is {\W[1]-hard} when parameterized by the leafage $\ell$.
Barnetson et al.~\cite{DBLP:journals/networks/BarnetsonBEHPR21} and Papadopoulos and Tzimas \cite{DBLP:conf/iwoca/PapadopoulosT22} presented \XP-algorithms running in time $n^{\calO(\ell)}$ for \textsc{Fire Break} and \textsc{Subset Feedback Vertex Set} 
on chordal graphs, respectively.
Papadopoulos and Tzimas \cite{DBLP:conf/iwoca/PapadopoulosT22} also 
proved that the latter problem is \W[1]-\hard\ when 
parameterized by the leafage.
Hochst{\"{a}}ttler et al~\cite{DBLP:journals/dam/HochstattlerHMP21} showed that we can compute the neighborhood polynomial of a chordal graph in $n^{\calO(\ell)}$-time.

For a chordal graph, leafage is a larger parameter than
the maximum clique-tree degree.
However, our next results imply that this parameter is not large enough
to obtain \FPT\ algorithms for either of the problems. 

\begin{restatable}{theorem}{thm:W-hardness-leafage-TS-Conn}
\label{thm:W-hardness-leafage-TS-Conn}
On chordal graphs, the 
\emph{\textsc{Token Sliding-Connectivity}} problem is 
\emph{co-\W[1]-\hard} when parameterized by the leafage $\ell$
of the input graph.
\end{restatable}

\begin{restatable}{theorem}{thm:W-hardness-leafage-TS-Reach}
\label{thm:W-hardness-leafage-TS-Reach}
On chordal graphs, the 
\emph{\textsc{Token Sliding-Reachability}} problem is 
\emph{\W[1]-\hard} when parameterized by the leafage $\ell$
of the input graph.
\end{restatable}

\noindent \textbf{Organization}
We use the standard notations that are specified in 
\cref{prelims}.
In Section~\ref{sec:max-clique-tree-degree} and 
Section~\ref{sec:leafage-hardness}, we present the hardness
results with respect to parameter maximum clique tree degree
and leafage, respectively.
We conclude the paper in Section~\ref{sec:conclusion}.

%% file: preliminaries.tex
\section{Preliminaries}
\label{prelims}

For a positive integer $q$, we denote the set $\{1, 2, \dots, q\}$ by $[q]$,
and for any $0 \leq p \leq q$, we denote the set $\{p,\dots,q\}$ by $[p,q]$.

\subparagraph*{Graph Theory.}
For a graph $G$, we denote by $V(G)$ and $E(G)$
the set of vertices and edges of $G$, respectively.
Unless specified otherwise, we use $n$ to denote the number of vertices in $G$.
{We denote the edge with endpoints $u, v$ by $uv$}.
For any $v \in V(G)$,
$N_G(v) = \{u \mid uv \in E(G)\}$ denotes the (open) neighbourhood of $v$,
and $N_G[v]=N_G (v) \cup \{v\}$ denotes the closed neighbourhood of $v$.
When the graph $G$ is clear from the context, we omit the subscript $G$.
For any $S \subseteq V(G)$, $G - S$ denotes the graph obtained from $G$ by deleting vertices in $S$.
We denote the subgraph of $G$ induced by $S$, i.e., the graph $G - (V(G) \setminus S)$, by $G[S]$.
We say graph $G$ contains graph $H$ as in \emph{induced subgraph}
if $H$ can be obtained from $G$ by a series of vertex-deletions. A {\em tree} $T$ is a connected acyclic graph. The sets $V_{\geq 3}(T)$ and $V_{=1}(T)$
denote the set of vertices of degree at least $3$,
and of degree equal to $1$, respectively.
The set $V_{\geq 3}(T)$ is also called the set of {\em branching vertices} of $T$
and the set $V_{=1}(T)$ is called the set of {\em leaves} of $T$.
Note that $|V_{\geq 3}(T)| \le |V_{=1}(T)|-1$.
Any node of $T$ which is not a leaf is called \emph{internal}.
For any further notation from basic graph theory, we refer the reader to~\cite{DBLP:books/daglib/Diestel12}.

\subparagraph*{Chordal graphs and Tree representations.}
A graph is called a chordal graph if it contains no induced cycle of length at least four.
It is well-known that chordal graphs can be represented as intersection graphs of subtrees in a tree, that is, 
for every chordal graph $G$, there exists a tree $T$ and a collection $\calM$ of subtrees of $T$ 
in one-to-one correspondence with $V(G)$ 
such that two vertices in $G$ are adjacent if and only if their corresponding subtrees intersect.

 
The \emph{leafage} of chordal graph $G$, denoted by $\ell$, is defined as 
the minimum number of leaves in the tree of a tree representation of $G$.
A tree representation $(T,\calM)$ for $G$ such that the number of leaves in $T$ is $\ell$, can be computed in time $O(|V(G)|^3)$ \cite{DBLP:conf/esa/HabibS09}. 

Let $G$ be a connected graph, and consider $\mathcal{K}_G$ to be the set of its maximal cliques. A tree on $\mathcal{K}_G$ is said to satisfy the \emph{clique-intersection} property if, for every pair of distinct cliques $K, K' \in \mathcal{K}_G$, the set $K\cap K'$ is contained in every clique on the path connecting $K$ and $K'$ in the tree. The tree $T^c$ for a chordal graph $G$ whose vertices are the cliques in $\mathcal{K}_G$ and which satisfies the clique-intersection property is called the \emph{clique-tree} of $G$. It is shown in \cite{10.1007/978-1-4613-8369-7_1} that the clique-tree $T_c$ for a chordal graph $G$ is isomorphic to its intersection tree $T$. Since in any tree the number of leaves is always greater than or equal to the maximum degree of a vertex, the leafage of a chordal graph $G$ is at least the 
maximum degree of its clique-tree $T^c$.

\subparagraph*{Parameterized Complexity.}
The input of a parameterized problem comprises an instance $I$,
which is an input of the classical instance of the problem,
and an integer $k$, which is called the parameter.
A parameterized problem $\Pi$ is said to be \emph{fixed-parameter tractable} (\FPT\ for short) %
if, given an instance $(I,k)$ of $\Pi$,
we can decide whether $(I,k)$ is a \yes-instance of $\Pi$
in time $f(k)\cdot |I|^{\OO(1)}$
for some computable function $f$ depending only on $k$.
A parameterized problem $\Pi$ is said to be \para\NPH\ %
if it is \NPH even for a constant value of $k$. 
For the purposes of this paper, we define a problem to be 
\WoneH if it is at least as hard as 
\textsc{MultiColored Independent Set (MultiCol Ind-Set)}
or \textsc{MultiColored Clique (MultiCol Clique)}.
In the first problem, an input is a graph $G$ with a 
partition $\langle V_1,V_2,\ldots,V_k\rangle$ of $V(G)$ such that each $V_i$ 
is an independent set and contains exactly $n$ vertices, and integer $k$.
The question is to determine whether $G$ contain 
an independent set of size $k$ that contains exactly one vertex from 
each partition.
The second problem is defined in a similar way.
If a problem is \WoneH, then it is unlikely to have an \FPT algorithm.
For more details on parameterized algorithms, and in particular parameterized branching algorithms,
we refer the reader to the book by Cygan et al.~\cite{DBLP:books/sp/CyganFKLMPPS15}.


%% file: reduction-bounded-clique-deg.tex

\label{sec:max-clique-tree-degree}

\subsection{Hardness for \textsc{Token Sliding-Connectivity}}

In this subsection, we prove Theorem~\ref{thm:np-hardness-clique-tree-degree-TS-Conn}.
We present a polynomial time reduction from the
\textsc{Dominating Set on Non-Blocking Graphs}
which is \NPH~\cite{DBLP:conf/wg/BonamyB17}.
For a graph $G$ and a subset $S$ of $V(G)$, a vertex $x\in V(G)$ is said to 
be a \emph{private neighbour} of $s\in S$ if $x\in N(s)$ and 
$x\notin N[t]$ for any $t\in S\setminus \{s\}$. 
A set of vertices $S\subset V(G)$ is a blocking set if no vertex in $S$ 
has a private neighbour with respect to $S$. 
In the \textsc{Dominating Set on Non-Blocking Graphs} problem,
an input is a graph $G$ and integer $k$ such that $G$ has no blocking set of size at most $2k-1$.
The objective is to determine whether $G$ contains a dominating set
of size at most $k$.

\noindent\textbf{Reduction:} 
Our reduction closely follows the one presented 
in~\cite{DBLP:conf/wg/BonamyB17}.
Let $(G,k)$ be an input instance of 
\textsc{Dominating Set on Non-Blocking Graphs}. 
The reduction constructs an instance 
$(G^\prime,k^\prime)$ of the \textsc{TS- Connectivity} problem.
Suppose $V(G) = \{v_1,v_2,\ldots,v_n\}$. 
\begin{itemize}
\item Add the vertices in the sets $C=\{c_1,c_2,\ldots,c_{n+k+1}\}$ 
and $W=\{w_1,w_2,\ldots,w_{n+k+2}\}$ in $V(G^\prime)$.
Add the edges $c_ic_j$ in $E(G^\prime)$ for $i\neq j\in [n+k+1]$
to ensure that $G^\prime[C]$ is a clique.
Set $W$ will remain an independent set in $G^{\prime}$.
\item For each $i,j\in [n]$, if $v_i$ and $v_j$ are adjacent in $G$ then add 
the edges $c_iw_j$ and $c_jw_i$ to $E(G^\prime)$.
\item Add the edge $c_iw_i$ in $E(G^\prime)$ for $i \in [n+k+1]$. 
\item Make vertex $w_{n+k+2}$ adjacent to all the vertices in $C$
that correspond to vertices in $V(G)$, i.e., 
add edges $c_iw_{n+k+2}$ in $E(G^\prime)$ for each $i\in [n]$.
\item Add sets of vertices $X=\{x_1,x_2,\ldots,x_{n+k+2}\}$ 
and $Y=\{y_1,y_2,\ldots,y_{n+k+2}\}$ in $V(G^\prime)$.
Also, add the edges $x_iy_i$ for $i\in[n+k+2]$ and edges 
$x_ic_j$ for each $i\in [n+k+2]$ and $j\in [n+k+1]$ in $G^\prime$.
That is, each vertex in $X$ is adjacent to all the vertices in $C$,
and $G^\prime[X \cup Y]$ consists of a matching as in 
\cref{TSC_deg_coNPH}. 
\end{itemize}
\begin{figure}[t]
    \centering
    \includegraphics[scale=0.35]{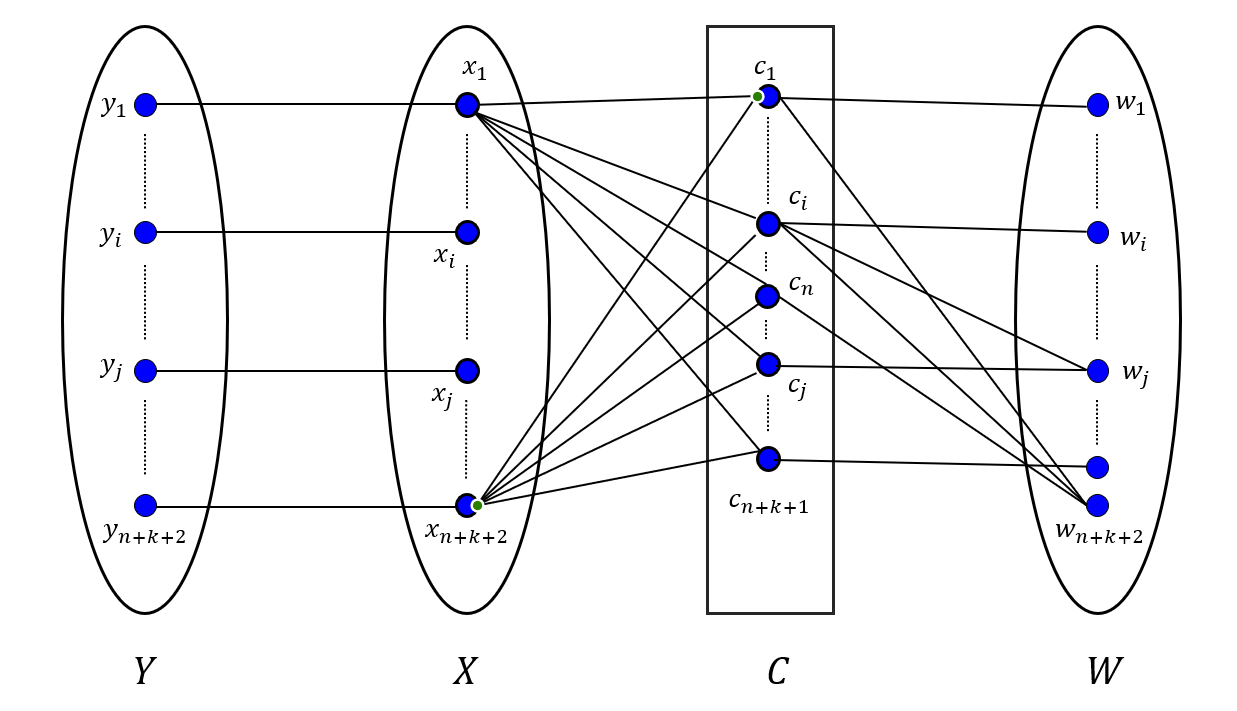}

    \includegraphics[scale=0.35]{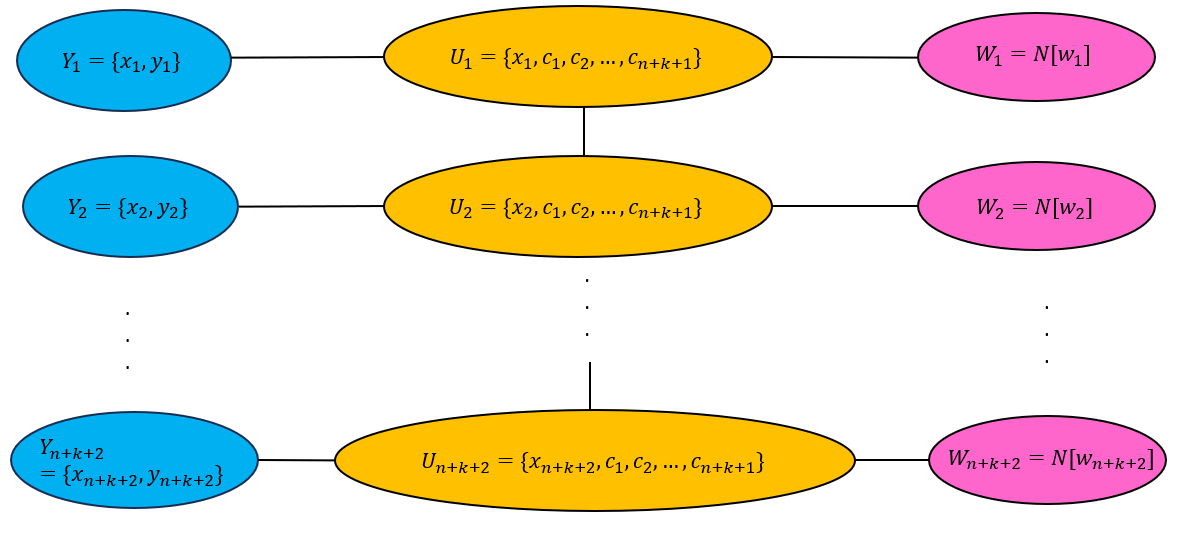}
    \caption{The reduced instance for the \textsc{Token Sliding Connectivity} problem parameterized by the maximum clique-tree degree and its corresponding clique-tree.}
    \label{TSC_deg_coNPH}
\end{figure}

The reduction sets $k^\prime=k+1$. 
This completes the description of the reduction.  
Next, we prove its correctness and graph $G^{\prime}$ satisfies
the desired properties.

\begin{lemma}
$G^\prime$ is a chordal graph with maximum clique-tree degree at most $4$.
\end{lemma}
\begin{proof}
Note that $G^\prime[C\cup X\cup W]$ is a split graph as 
$C$ induces a clique and $X\cup W$ induces an independent set 
in $G^\prime$. 
The graph $G^\prime$ is formed by adding a pendant vertex 
to each vertex in $X$ in the graph $G^\prime[C\cup X\cup W]$;
hence the graph $G^\prime$ is a chordal graph. 

We now construct a clique-tree $\calT_c$ for $G^\prime$, 
and show that it has a maximum degree of $4$. 
Add a vertex $U_i$ in $\calT_c$ corresponding to the clique 
induced by the vertices $\{x_i,c_1,c_2,\ldots,c_{n+k+1}\}$ 
in $G^\prime$ for $i \in [n+k+2]$. 
Similarly, add a vertex $W_i$ in $\calT_c$ corresponding 
to the clique induced by the vertices in $N[w_i]$ in $G^\prime$ 
for $i\in [n+k+2]$. 
Add a vertex $Y_i$ corresponding to the clique (i.e., edge) induced 
by the vertices $\{x_i,y_i\}$ for $i\in[n+k+2]$. 
Add an edge from $U_i$ to $U_{i+1}$ for each $i\in [n+k+1]$. 
Add an edge from $U_i$ to $W_i$ for each $i\in [n+k+2]$.
Add an edge from $U_i$ to $Y_i$ for each $i\in [n+k+2]$. 
Thus, the maximum degree of $\calT_c$ is $4$. 
Each vertex in $V(\calT_c)$ corresponds to a maximal clique in $G^\prime$.
This tree is also shown in \cref{TSC_deg_coNPH}.

Next, we prove that this tree satisfies the clique-intersection property.
Each vertex in $Y\cup W$ appears in only one vertex (bag) of $\calT_c$.
Each vertex $x_i\in X$, for some $i \in [n+k+2]$, 
appears in the vertex (bag) $Y_i$ and all the vertices $U_j$ 
where $j\in [n+k+2]$. 
Since the vertices $U_j$ where $j\in [n+k+2]$ induce a path in 
$\calT_c$ and the vertex $Y_i$ is adjacent to the vertex $U_i$;
the set of vertices containing $x_i$ induces a connected sub-tree of 
$\calT_c$. 
Now any vertex $c_i\in C$ (for some $i\in [n+k+1]$) appears 
in all the bags $U_j$ where $j\in [n+k+2]$, and possibly some 
of the bags $W_\ell$ where $\ell\in [n+k+2]$.
This induces a path with some pendant vertices, 
which is a connected subtree of $\calT_c$. 
Thus, $G^\prime$ has a clique-tree with maximum degree $4$.
\end{proof}
\begin{lemma}
\label{lemma:TS-conn-max-deg}
$(G,k)$ is a \yes\ instance of the
\emph{\textsc{Dominating Set on Non-Blocking Graphs}} 
problem if and only if $(G^\prime,k^\prime)$ is a \no\ instance of the
\emph{\textsc{TS-Connectivity}} problem. 
\end{lemma}
\begin{proof}
Suppose $(G,k)$ is a \yes\ instance of the
\textsc{Dominating Set on Non-Blocking Graphs} problem. 
Without loss of generality, let $D=\{v_1,v_2,\ldots,v_k\}$
be a dominating set in $G$. 
We construct an independent set $I$ of size $k+1$ in $G^\prime$ 
such that no token on $I$ can be moved. 
Define $I :=\{w_1,w_2,\ldots,w_k,w_{n+k+2}\}$. 
Clearly, the token on $w_{n+k+2}$ cannot be moved because 
$N(w_{n+k+2})=\{c_i \mid i\in [n]\}$ and 
if we move this token to $c_j$ where $j\in [n]$, 
it will be adjacent to the token on $w_\ell$ for some 
$\ell \in [k]$ as $v_j\in N[v_\ell]$ for some $\ell\in [k]$ as $D$
is a dominating set in $G$. 
Now, if we move the token on $w_i$ for some $i\in [k]$ to any vertex 
in $N(w_i)\subset \{c_1,c_2,\ldots,c_n\}$, this will be adjacent to the 
token on $w_{n+k+2}$ because $N(w_{n+k+2})= \{c_1,c_2,\ldots,c_n\}$.
Thus $I$ will form an isolated vertex in the configuration graph 
$TS_{k+1}(G^\prime)$. 
As there are another independent sets of size $k+1$ in $G^\prime$,
such as $\{w_{n+1}, w_{n+2}, \dots, w_{n+k+1}\}$, 
$(G^{\prime}, k)$ is a \no\ instance of the \textsc{TS-Connectivity} problem. 

Suppose we have a \no\ instance of the
\textsc{Dominating Set on Non-Blocking Graphs} problem, 
then we show that $(G^\prime,k+1)$ is a \yes\ instance of the
\textsc{TS-Connectivity} problem. 
Define independent set $J:=\{w_{n+1},w_{n+2},\ldots,w_{n+k+1}\}$.
We show that any independent set $I$ of size $k+1$ in 
$G^\prime$ can be reconfigured to $J$ using a sequence 
of valid token sliding moves. 
It is sufficient to show that the independent set $I$ can be reconfigured 
into an independent set $I^\prime$ such that $|I^\prime\cap J|>|I\cap J|$. 
We will demonstrate this through a case analysis depending 
on how $I$ intersects the partitions $C, X, Y$, and $W$.

\emph{Case $(I)$: $I\cap C\neq \emptyset$:}
Let $\{c_i\}=|I\cap C|$ as $I$ is an independent set and 
hence it cannot contain two or more vertices from $C$. 
Now the other tokens in $I$ can only be in $Y\cup W$. 
If $i>n$, move the token on $c_i$ to $w_i$. 
This will give the desired independent set $I^\prime$. 
If $i\leq n$, bring the token on $I$ to $c_j$ where $j\in [n+1,n+k+1]$
and $w_j$ has no token. 
This is a valid move since $c_j$ has no neighbours in $Y\cup W$ 
other than $w_j$. 
Now move this token from $c_j$ to $w_j$ to get the desired 
independent set $I^\prime$.

\emph{Case $(II)$: $I\cap C=\emptyset$ but $I\cap X\neq \emptyset$:}
Let $x_i\in I\cap X$ for some $i\in [n+k+2]$. 
If there are more than one tokens in $I$ on the vertices in $X$, 
slide the tokens on vertices $x_j$ where $j\neq i, j\in [n+k+2]$ to $y_j$. 
There can be at most $k$ tokens in $W$ 
as there is at least one token in $X$.
Now bring the token on $x_i$ to $c_\ell$ where $\ell\in [n+1,n+k+1]$ and there is no token on $w_\ell$. This is a valid move since $c_\ell$ has no neighbours in $Y\cup W$ other than $w_\ell$. Now bring this token from $c_\ell$ to $w_\ell$ to get the desired independent set $I^\prime$.

\emph{Case $(III)$: $I\cap C=I\cap X=\emptyset$ but $I\cap Y\neq \emptyset$:} 
Suppose $y_i\in I \cap Y$. Bring the token on $y_i$ to $x_i$. This move is valid because $X\cup C$ has no token. Now proceed as in Case~$(II)$.

\emph{Case $(IV)$: $I\cap C=I\cap X=I\cap Y=\emptyset$ and $w_{n+k+2}\in I\cap W$:}
Consider the set $D_1=\{v_i| \,w_i~\text{has a token and}~ i\in [n]\}$.
The cardinality of this set is at most $k$.
The set $D_1\subset V(G)$ cannot be a dominating set as $G$ does not have a dominating set of size at most $k$. 
Therefore there exists $v_j$ for some $j\in [n]$ such that ${N[v_j]\cap D_1}=\emptyset$. 
The vertex $c_j$ is not be adjacent to any vertex in $W\cap I$ 
except $w_{n+k+2}$. 
Move the token on $w_{n+k+2}$ to $c_j$ and proceed as in Case~$(I)$.

\emph{Case $(V)$: $I\cap C=I\cap X=I\cap Y=\emptyset$ and $w_{n+k+2}\notin I\cap W$:} 
Consider the set $D_2=\{v_i| \,w_i~\text{has a token and}~i\in [n]\}$. 
It will have size at most $k+1$. 
The set $D_2 \subset V(G)$ cannot be a blocking set as 
$G$ does not have a blocking set of size at most $2k-1$. 
Therefore, there exists $v_j\in D_2$ for some $j\in [n]$ 
such that $v_j$ has a private neighbour $v_\ell$ for some $\ell\in[n]$. 
Move the token on $w_j$ to $c_\ell$. 
This will not be adjacent to any other token in $W$. 
Now proceed as in Case~$(I)$.

Hence, if $(G,k)$ is a \no\ instance of the \textsc{Dominating Set} problem, then $(G^\prime,k)$ is a \yes\ instance of the \textsc{TS-Connectivity} problem.
This completes the proof of the claim.
\end{proof}

These arguments, along with the fact that reduction 
takes polynomial-time, imply that \textsc{TS-Connectivity} is {\co-\NP-hard} on chordal graphs with a maximum clique-tree degree of $4$.

\subsection{Hardness for \textsc{Token Sliding Reachability}}
In this subsection, we present a reduction used to prove \Cref{thm:np-hardness-clique-tree-degree-TS-Reach}.
We present a polynomial time reduction from
\textsc{TS-Reachability} on split graphs
which is known to be \NPH~\cite{DBLP:journals/mst/BelmonteKLMOS21}.

 \begin{figure}
    \centering
    \includegraphics[width=5cm]{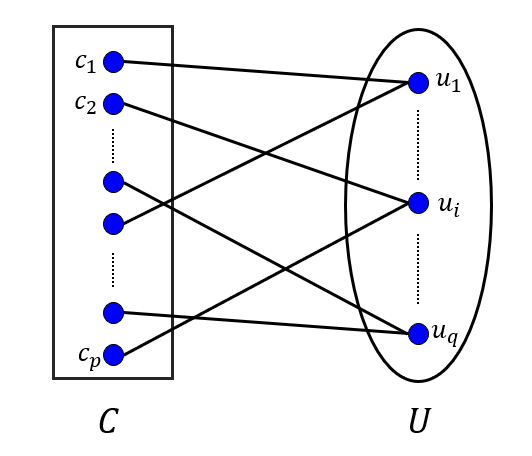}
    \caption{The input Instance of \textsc{TS-Reachability}, $G$ is a split and $C$ is a clique graph.}
    \label{TSR_deg}
\end{figure}

\textbf{Reduction}
Let $(G,I,J)$ be the input instance of \textsc{TS-Reachability} 
problem where $G$ is a split graph with a clique $C=\{c_1,c_2,\ldots,c_p\}$ 
and an independent set $U=\{u_1,u_2,\ldots,u_q\}$ and
$I,J$ are two $k$-independent sets in $G$.
The reduction constructs an instance 
$(G^\prime,I^\prime,J^\prime)$ of the \textsc{TS-Reachability} problem
as follows.
\begin{itemize}
\item For each vertex $c_i\in V(G)$ construct a vertex $d_i$, 
and for each vertex $u_j\in V(G)$ construct a vertex 
$w_j$ where $i\in[p]$ and $j\in[q]$.
Add edges between all pairs of vertices $d_i$ and $d_j$ where 
$i\neq j\in [p]$ to construct the clique $C^{\prime}$.
Denote the set of vertices $\{w_1,w_2,\ldots,w_q\}$ by $W$. 
This will be an independent set in $G^\prime$.
\item For each edge $c_iu_j$ in $E(G)$ (where $i\in [p]$ and $j\in [q]$), add an edge $d_iw_j$.
\item Construct $q$ vertices $s_1,s_2,\ldots,s_q$. We will denote the set $\{s_1,s_2,\ldots,s_q\}$ by $S$ and refer to the vertices in $S$ as \emph{special} vertices. Now add an edge from each $s_i$ to each $d_j$ such that $i\in [q]$ and $j \in [p]$. That is, each vertex $s_i$ in $S$ should be adjacent to all the vertices $d_j$ in $C^\prime$, as shown in \cref{TSR_deg_NPH}.
\end{itemize}
\begin{figure}
    \centering
    \includegraphics[scale=0.40]{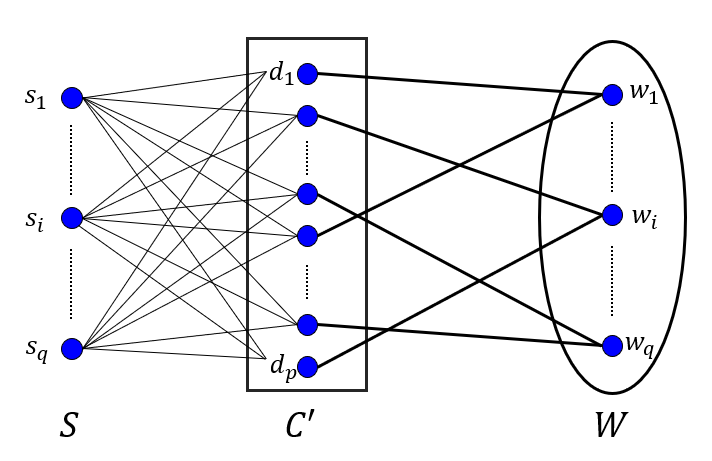}
    \includegraphics[scale=0.3]{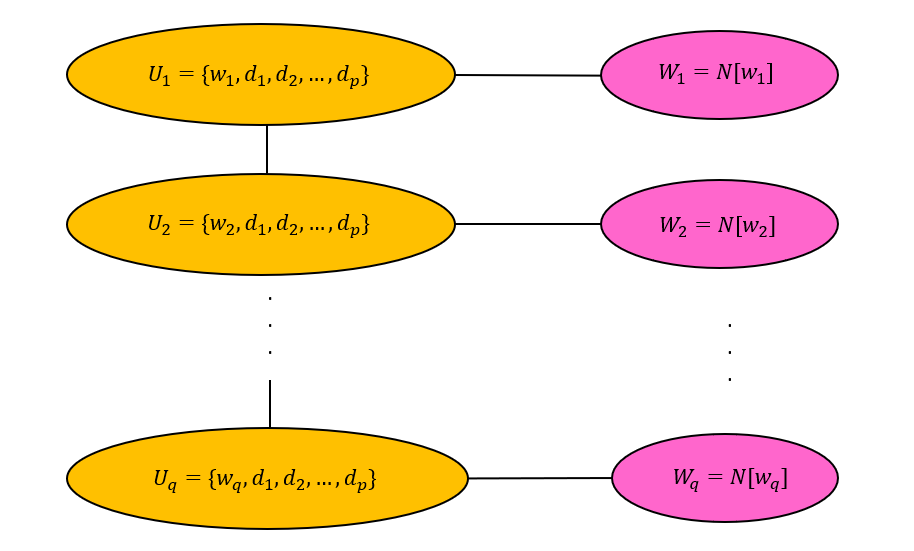}
    \caption{Reduction for \textsc{Token Sliding Reachability} problem parameterized by the maximum clique-tree degree and the corresponding clique-tree.}
    \label{TSR_deg_NPH}
\end{figure}
This completes the description of the reduced instance. 
\begin{lemma}
$G^\prime$ is a chordal graph with a maximum clique-tree degree of
at most $3$.
\end{lemma}
\begin{proof}
Clearly, $G^\prime$ is a split graph with clique $C^\prime$ 
and independent set $S\cup W$. 
Hence $G^\prime$ is chordal. 
We construct a clique-tree $\calT_c$ for $G^\prime$ as follows: 
Add a vertex $U_i$ in $\calT_c$ corresponding to the clique 
induced by the vertices $\{w_i,d_1,d_2,\ldots,d_p\}$ in $G^\prime$ 
for $i\in [q]$. 
Similarly, add a vertex $W_i$ in $\calT_c$ corresponding 
to the clique induced by the vertices in $N[w_i]$ in $G^\prime$ for $i\in [q]$. 
Add an edge from $U_i$ to $U_{i+1}$ for each $i\in [q-1]$. 
Add an edge from $U_i$ to $W_i$ for each $i\in [q]$. 
Thus, the maximum degree of $\calT_c$ is $3$. 
Each vertex in $V(\calT_c)$ corresponds to a maximal clique in 
$G^\prime$. 
Each vertex in $S\cup W$ appears in just one vertex (bag) of $\calT_c$,
and each vertex in $C^\prime$ appears in all the vertices $U_i$ for 
$i\in [q]$ and possibly some of the vertices $W_j$ for $j\in [q]$, 
which induces a path in $\calT_c$ with possibly some pendant edges.
Thus, for each vertex $v\in V(G^\prime)$, the set of vertices (bags) of 
$\calT_c$ which contain $v$, 
induces a connected sub-tree of $\calT_c$. 
Hence, $G^\prime$ has a clique-tree of maximum degree $3$. 
\end{proof}
\begin{lemma}
$(G,I,J)$ is a \yes\ instance of the \textsc{Token Sliding Reachability} problem
if and only if $(G^\prime, I^\prime, J^\prime)$ is a \yes\ instance of the \textsc{Token Sliding Reachability} problem.
\end{lemma}
\begin{proof}
The subgraph induced by $C^\prime \cup W$ in $G^\prime$ 
is isomorphic to $G$ where the isomorphism $\phi$ is given by 
$\phi(c_i)=d_i$ for $i\in [p]$ and $\phi(u_j)=w_j$ for $j \in [q]$.
Define $\phi(T)=\bigcup_{v\in T}{\phi(v)}$ for any subset $T$ of $G$. 
Define $I^\prime=\phi(I)$ and $J^\prime=\phi(J)$. 
Since $I$ and $J$ are independent sets in $G$, $I^\prime$ and $J^\prime$ are independent sets in $G^\prime$.

In the forward direction,
it is sufficient to show that if any independent set $I_1$ in $G$ 
can be reconfigured to an independent set $I_2$ in $G$ in 
one valid move of token sliding, then the independent set 
$\phi(I_1)$ in $G^\prime$ can be reconfigured to an 
independent set $\phi(I_2)$ in $G$ in one valid move of token sliding.
If $I_2=I_1\setminus \{u\}\cup \{v\}$, then $\phi(I_2)=\phi(I_1)\setminus \{\phi(u)\}\cup \{\phi(v)\}$. 
Thus $\phi(I_2)$ can be reconfigured from $\phi(I_1)$ by moving 
the token on $\phi(u)$ to $\phi(v)$. 
Since $\phi$ is an isomorphism, this is a valid move.

In the reverse direction, if $I^\prime=J^\prime$, 
then $I=J$ and we trivially have a \yes\ instance.
Therefore, assume $I^\prime \neq J^\prime$ and thus at least 
one token needs to be moved to reconfigure $I^\prime$ to $J^\prime$.
If at any intermediate step, a special vertex contains a token, 
then there cannot be a token in $C^\prime$ and thus no other 
token can move. 
Therefore we will need to move this token to $C^\prime$. 
Thus any sequence of valid token sliding moves in $G^\prime$, 
where the initial and final independent sets do not contain 
a special vertex, can be modified to a sequence of valid token 
sliding moves in $G^\prime$ where no intermediate configuration 
has a token on a special vertex, with the same initial and final independent 
sets.
Hence we can get a sequence of valid token sliding moves from $I^\prime$ 
to $J^\prime$ which places tokens only on vertices of $C^\prime \cup W$. 
Thus, we can get a sequence of independent sets obtained by valid token
sliding moves in $G$ by replacing each intermediate independent set 
$K$ in this sequence by $\phi^{-1}(K)$.
\end{proof}

The above two claims, along with the fact that reduction takes polynomial-time,
imply that \textsc{Token Sliding Reachability} is \NP-\hard\ on chordal graphs with maximum clique-tree degree $3$.


%% file: TS-conn-W1hard.tex
\subsection{Hardness for \textsc{Token Sliding Connectivity}}

In this subsection, we prove \Cref{thm:W-hardness-leafage-TS-Conn}.
We present a parameter-preserving reduction that takes as input an 
instance $(G,\langle V_1,V_2,\ldots,V_k\rangle, k)$ of 
\textsc{MultiCol Ind-Set} and returns an instance 
$(G^\prime,nk)$ of \textsc{TS-Connectivity} 
where $G^\prime$ is a chordal graph with leafage $\ell = 2 \cdot k$. 
We find it convenient to describe a tree model $\mathcal{T}$ of the 
chordal graph $G^\prime$.
Recall that in this model,
each vertex of $G^\prime$ corresponds to a specified subtree of 
$\mathcal{T}$, and two vertices of $G^\prime$ are adjacent 
if and only if their corresponding subtrees have a non-empty intersection.

\subparagraph*{Structure of the model tree $\mathcal{T}$, Parking Structure, and Conditional Free Pass }

The model tree $\mathcal{T}$ consists of a central vertex $t_0$, which is the root and has $k+1$ children $t_1,t_2,\ldots,t_k$, and $t_p$.
For each $i\in[k]$, $t_i$ has two children $t_i^a$ and $t_i^b$.
The reduction subdivides each edge $t_i t_i^a$ and $t_i t_i^b$ of the tree $\mathcal{T}$ by adding $2n-1$ new vertices on this edge.
Denote the path between {$t_i^a$} and $t_i^b$ by $T_i$ for each $i \in [k]$.
We use $T_i$ to encode vertices in $V_i$.
It also subdivides $t_0 t_p$ and uses it to park the tokens as
described below.
\begin{itemize}
\item {We subdivide $t_0t_p$ of the tree $\mathcal{T}$ by adding} $nk$ blue vertices and $nk-1$
green vertices as shown in \Cref{structure}.
This allows us to park $nk$ tokens on blue vertices in this `parking structure'.
\item We add a purple vertex $b^\star$ in $G$ whose model contains
${t_0,t_1,t_2,\ldots,t_k}$ and the first vertex in the path from
$t_0$ to $t_p$, i.e., it is a star centered at $t_0$ with $k + 1$ leaves.
\end{itemize}
Vertex $b^{\star}$ will act as a bridge and allow us to take
tokens from any $t_i$ (for some $i\in [k]$) to the parking spot
provided it is possible to move a token to $t_i$.

Next, we move to the conditional free pass.
The idea is to ensure that if there is no token in $T_i$
(which is the condition), then all the tokens in $T_j$
can be moved to $t_j$, then to $b^{\star}$ and eventually
to the parking structure.
For each $i\neq j$ where $i,j\in [k]$, we add a orange vertex that has the following model.
\begin{itemize}
\item In $T_i$, add a line segment from $t_i^a$ to $t_i^b$.
In $T_j$, add a line segment from the first vertex from $t_j$ towards
$t_j^a$ to the first vertex from $t_j$ towards $t_j^b$.
Finally, connect these two line segments by an inverted $\calV$-like
structure with the top of the inverted $\calV$ corresponding to the vertex
$t_0$ as shown in \Cref{structure}.
\end{itemize}

\begin{figure}[t]
    \centering
    \includegraphics[scale=0.4]{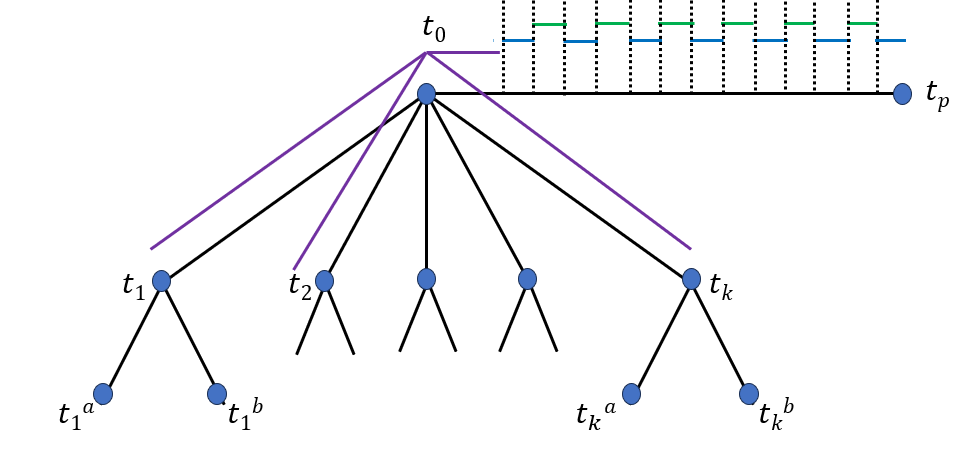}
    \includegraphics[scale=0.5]{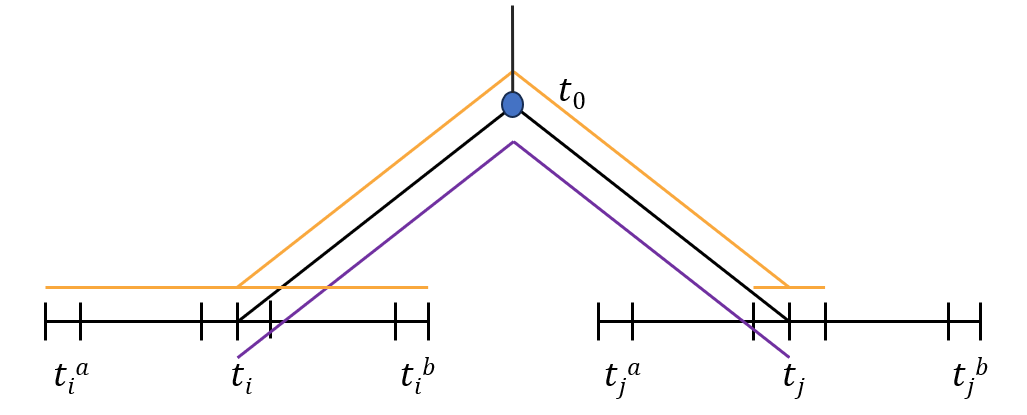}
    \caption{(Top) Structure of the model tree. Node $t_0$ is the central vertex and $t_1, t_2, \dots, t_k$ and $t_p$ are its children.
     The structure between $t_0$ and $t_p$ is to park the tokens.
     (Bottom) The orange vertex denotes the conditional free pass
     between $T_i$ and $T_j$.
     \label{structure}}
     \vspace{-5mm}
\end{figure}

\subparagraph*{Encoding the vertices of $G$:}
We add the following three types of vertices in $G^{\prime}$ to encode
$n$ vertices in $V_i$ for each $i \in [k]$.
Recall that we have subdivided the edge $t_it_i^a$ and the edge $t_it_i^b$
of the tree $\mathcal{T}$ and added $2n-1$ new vertices on each of
these edges.
\begin{itemize}
\item
Add $n$ new vertices $u_i^1,u_i^2,\ldots,u_i^n$ in $G^\prime$
corresponding to the $n$ disjoint intervals on this edge starting
from the vertex $t_i^a$ 
to the first new vertex, from the second new vertex to the third new vertex and so on as shown in
\Cref{Encoding_Vertices_TSC}. Here, the vertex $u_1$ corresponds to the interval between the vertex $t_i^a$ and the first new vertex and so on.
Also, add $n-1$ connector vertices that
connect $u_i^p$ and $u_i^{p+1}$ where $p\in [n]$.
\item
Similarly, add $n$ new vertices $w_i^1,w_i^2,\ldots,w_i^n$ in $G^\prime$
corresponding to the $n$ disjoint intervals on this edge starting
from the first new vertex to the second new vertex and so on such that $w_i^n$
corresponds to the interval from the last new vertex to the vertex $t_i^b$ as shown in \Cref{Encoding_Vertices_TSC}.
As before, add $n-1$  connector vertices for each
$i\in[k]$ which connect $w_i^p$ and $w_i^{p+1}$.
\item As shown in \Cref{Encoding_Vertices_TSC}, add $(n-1)$ pink vertices
for each $t_i$ where $i\in [k]$ as follows:
Each pink vertex $y_i^p$ intersects ${u_i^p,u_i^{p+1},\ldots,u_i^n}$
and ${w_i^1,w_i^2,\ldots,w_i^{p+1}}$ for $p\in [n-1]$.
\end{itemize}

Note that if there are $p$ tokens on the left-hand side of $T_i$
(i.e., between $t_i^a$ and $t_i$) it is safe to assume that they are
on $u_i^1,u_i^2,\ldots,u_i^p$.
If not, then we can move them to $u_i^1,u_i^2,\ldots,u_i^p$ by valid
token sliding operations without disturbing the tokens on
the rest of the graph.
Similarly, if there are $p$ tokens on the right-hand side of $T_i$
(i.e., between $t_i^b$ and $t_i$) it is safe to assume that they are on
$w_i^{n-p+1},w_i^{n-p+2},\ldots,w_i^n$.
The pink vertices are added to ensure the following two
claims hold.
\begin{restatable}{claim}{tsConnClaimNTokens}
\label{cl:ntokens}
For any $i \in [k]$, if there are $n$ tokens on $T_i$, then no token
can be moved to a pink vertex.
\end{restatable}
\begin{claimproof}
We first consider the case when all the tokens are on
${w_i^1,w_i^2,\ldots,w_i^n}$.
If we move the token on some $w_i^q$ to some pink vertex $y_i^r$
where $r\neq q$, it will be adjacent to the token on the vertex $w_i^r$.
If we move the token on some $w_i^q$ to $y_i^q$ then it
will be adjacent to the token on some $w_i^r$ where $r\neq q$.
Hence, in this case, no token can be moved to a pink vertex.
Using similar arguments, if all tokens are
on ${u_i^1,u_i^2,\ldots,u_i^n}$, then no token can be moved to a pink vertex.

Consider the case when tokens are on ${u_i^1,u_i^2,\ldots,u_i^p}$
and on ${w_i^{n-p+1},w_i^{n-p+2},\ldots,w_i^n}$
for some $0<p<n$.
Suppose we move the token on $u_i^1$ to the pink vertex $y_i^1$,
this token will be adjacent to the token on $u_i^p$ if $p\geq2$
otherwise it will be adjacent to the token on $w_i^2$.
If we move the token on $u_i^q$ where $2\leq q \leq p$
to a pink vertex $y_i^r$ where $r\leq q$,
then it will be adjacent to $u_i^1$.
If we move the token on $w_i^q$ to a pink vertex $y_i^r$
such that $r\geq q\geq p+2$, then it will be adjacent to
the token on $w_i^{p+1}$.
If $p\neq n-1$ and we move the token on $w_i^{p+1}$
to a pink vertex $y_i^r$ such that $r\neq p+1$,
then it will be adjacent to the token on $w_i^n$.
If we move the token on $w_i^{p+1}$ to the pink vertex
$y_i^{p+1}$, then it will be adjacent to the token on the vertex $u_i^p$.
Thus, no token from $T_i$ can be moved to a pink vertex via a valid token sliding move.
\end{claimproof}
\begin{restatable}{claim}{tsConnNrTokenEachSide}
\label{cl:lessthann}
For any $i \in [k]$, if there are fewer than $n$ tokens on $T_i$, then
the pink vertices intersecting $T_i$ can be used to move tokens to $b^{\star}$ (and then to the parking structure).
\end{restatable}
\begin{claimproof}
As before, we remark that the claim only implies the condition on the
usability of the pink vertices intersecting $T_i$.
It might be possible that $b^{\star}$ is adjacent to some other
token and hence the tokens in $T_i$ cannot be moved to it.
However, for the sake of clarity, we assume that this is not the case
in the rest of the proof.

Without loss of generality, we assume that there are tokens on
$u_i^1,u_i^2,\ldots,u_i^p$ and $v_i^{n-q+1},v_i^{n-q+2},\ldots,v_i^n$
where $p+q<n$.
Move the token on $u_i^p$ to $y_i^p$ which is not adjacent
to any other token.
Now move this token to the purple vertex $b^{\star}$
and then to the last vertex of the parking space.
Similarly move each $u_j$ (where $j<p$ and $j$ is the index
of the largest remaining token between $t_i^a$ and $t_i$) to $b^{\star}$.
Then move it to the last empty blue vertex.
Thus all the tokens on $u_i^1,u_i^2,\ldots,u_i^p$ can be moved
to $b^{\star}$.
Similarly all the tokens on $v_i^{n-q+1},v_i^{n-q+2},\ldots,v_i^n$
can be moved to the parking space via $b^{\star}$ by starting
from the token on the vertex with the smallest index.
\end{claimproof}

\begin{figure}[t]
    \centering
    \includegraphics[scale=0.5]{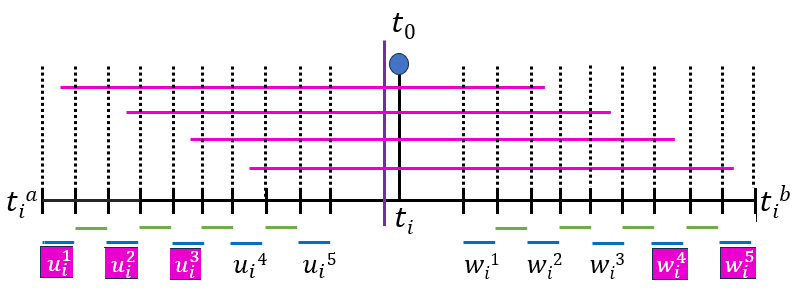}
    \caption{Add $(n-1)$ pink vertices for each $t_i$ where $i\in[k]$ as follows: Each pink vertex $y_i^p$ intersects $\{u_i^p,u_i^{p+1},\ldots,u_i^n\}$ and $\{w_i^1,w_i^2,\ldots,w_i^{p+1}\}$. In the figure $n$ is taken to be $5$.
    No token from $T_i$ can be moved to a pink vertex.}
    \label{Encoding_Vertices_TSC}
    \vspace{-5mm}
\end{figure}

\subparagraph*{Encoding the edges of $G$:}
Whenever there is an edge from the $p^{th}$ vertex in $V_i$ to the $q^{th}$ vertex in $V_j$ where $p,q\in[n]$ and $i,j\in [k]$, add a vertex in $G^\prime$ as described below. Refer \Cref{Encoding_edges_TSC} for an illustration.
\begin{itemize}
\item Add a horizontal line segment from the endpoint of the interval corresponding to $u_i^p$ to the endpoint of the interval corresponding to $w_i^p$ (both inclusive). Add another horizontal line segment from the endpoint of the interval corresponding to $u_j^{q+1}$ to the starting point of the interval corresponding to $w_j^q$ (both inclusive). Join these two horizontal line segments by a vertical line segment.
\item Similarly, add a horizontal line segment from the endpoint of the interval corresponding to $u_j^q$ to the starting point of the interval corresponding to $w_j^q$ (both inclusive). Add another horizontal line segment from the endpoint of the interval corresponding to $u_i^{p+1}$ to the starting point of the interval corresponding to $w_i^p$. Join these two horizontal line segments by a vertical line segment.
\end{itemize}
We denote the vertices of $G^\prime$ which are added corresponding to the edges of $G$ by $\calH$-type vertices as we add two structures similar to letter $H$.

If there are $n$ tokens in $T_i$ for some $i$, none of the tokens in $T_i$
can be moved to the parking structure using the pink vertices as shown in 
 \Cref{cl:ntokens}.
Thus in this case we will need to use the $\calH$-type vertices to move at least one of these tokens to the parking space.

\begin{restatable}{claim}{tsConnclEdge}
\label{cl:edge}
Consider an edge $e$ of $G$ from the $p^{th}$ vertex in $V_i$ 
to the $q^{th}$ vertex in $V_j$ where $p,q\in [n]$ and $i,j\in [k]$.
Suppose $n$ tokens in $T_i$ are placed on $\{u_i^1,u_i^2,\ldots,u_i^p\}$ 
and $\{w_i^{p+1},w_i^{p+2},\ldots,w_i^n\}$. 
A token in $T_i$ can be moved to $b^{\star}$ and then to the parking 
space using the $\calH$-type structure corresponding to $e$ 
if and only if $T_j$ contains $n$ tokens placed on 
$\{u_j^1,u_j^2,\ldots,u_j^q\}$ and 
$\{w_j^{q+1},w_j^{q+2},\ldots,w_j^n\}$.
\end{restatable}
\begin{claimproof}
Suppose the $n$ tokens in $T_j$ are on $\{u_j^1,u_j^2,\ldots,u_j^q\}$ 
and $\{w_j^{q+1},w_j^{q+2},\ldots,w_j^n\}$. 
None of these are adjacent to the $\calH$-type structure 
corresponding to $e$. 
The only token on $T_i$ which is adjacent to the 
$\calH$-type structure corresponding to $e$ is the 
token on $u_i^p$, hence it can move to the 
$\calH$-type structure without violating the independence.

Now we show the converse. Without loss of generality, 
suppose the $n$ tokens in $T_j$ are on 
$\{u_j^1,u_j^2,\ldots,u_j^r\}$ and $\{w_j^{r+1},w_j^{r+2},\ldots,w_j^n\}$ 
where $r\neq q$. 
If $r>q$, then the token on $u_j^r$ will 
be adjacent to the $\calH$- type structure. 
Similarly if $r<q$, the token on $w_j^q$ 
will be adjacent to the $\calH$-type structure. 
Thus, no token in $T_i$ can be moved to the $\calH$-type structure.
\end{claimproof}

\begin{figure}
    \centering
    \includegraphics[scale=.55]{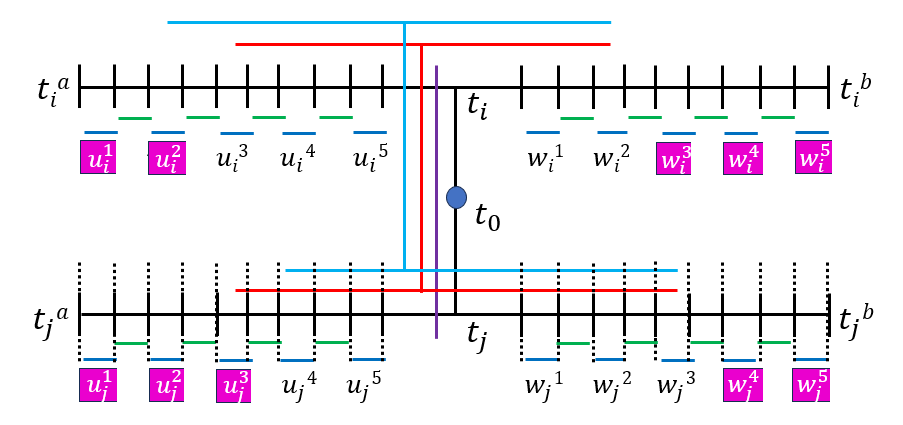}
    \caption{For $n = 5$, if there exists an edge between the $2^{nd}$
    vertex of $V_i$ and $3^{rd}$ vertex of $V_j$, 
    this is the $\calH$-type structure corresponding to this edge.
    This is the $\calH$-type structure when $n=5$ and 
    there is an edge between $2^{nd}$ vertex of $V_i$ and 
    the $3^{rd}$ vertex of $V_j$. 
    Notice that one of the tokens on $u_i^{2}$ (or $u_j^{3}$) 
    can move to $t_0$ using the cyan (or the red) vertex.}
\label{Encoding_edges_TSC} 
\vspace{-5mm}
\end{figure}

This completes the description of the reduction and the necessarily claims.
Next, we show that $(G,\langle V_1,V_2,\ldots,V_k\rangle, k)$
is a \yes\ instance of \textsc{MultiCol Ind-Set} if and only if
the instance $(G^\prime,nk)$ is a \no\ instance of 
\textsc{TS-Connectivity}. 
Let the independent set in $G^\prime$ which consists of the 
$nk$ vertices in the parking structure be denoted by $I^\star$.
We use the fact that $(G^\prime,nk)$ is a \yes\ instance of 
\textsc{TS-Connectivity} if and only if any 
independent set $I$ of size $nk$ in $G^\prime$ can be 
modified to $I^\star$ via a sequence of valid token sliding operations.

\begin{restatable}{lemma}{tsConnforward}
\label{lemma:forward-reduction-leafage}
If $(G,\langle V_1,V_2,\ldots,V_k\rangle, k)$ is a \yes\ instance of \textsc{MultiCol Ind-Set},
then $(G^\prime,nk)$ is a \no\ instance of \textsc{TS-Connectivity}.
\end{restatable}
\begin{proof}
Supose $X=\{v_1^{p_1},v_2^{p_2},\ldots,v_k^{p_k}\}$ be a solution of 
\textsc{MultiCol Ind-Set} where the $p_i^{th}$ vertex from 
partition $V_i$ is in the solution for every $i\in [k]$. 
We construct an independent set $I$ in $G^\prime$ of size 
$nk$ as follows:
For each $i \in [k]$, place $n$ tokens in each $T_i$ on the vertices 
$u_i^1,u_i^2,\ldots,u_i^{p_i}$ and $w_i^{p_i+1},w_i^{p_i+2},\ldots,w_i^n$. 
We show that $I^\star$ is not reachable from $I$ via a set 
of valid token sliding moves.
Since there are $n$ tokens in $T_i$, pink vertices cannot 
be used to move a token in $T_i$ to the parking structure by 
\Cref{cl:ntokens}.
Moreover, by the construction, the orange vertex intersecting both 
$T_i$ and $T_j$ can be used to move tokens in $T_j$ to $b^{\star}$
(and then to the parking structure) if and only if 
$T_i$ does not contain any token.
And hence, the orange vertex also cannot be used 
to move a token in $T_i$ to the parking structure.
Hence it will be possible to move a token to the parking structure 
only using the $\calH$-type structure. 
Using \Cref{cl:edge}, when there are $n$ tokens each in both 
$T_i$ and $T_j$ such that the tokens in $T_i$ are at 
$\{u_i^1,u_i^2,\ldots,u_i^p\}$ and 
$\{w_i^{p+1},w_i^{p+2},\ldots,w_i^n\}$ and 
the tokens in $T_j$ are at $\{u_j^1,u_j^2,\ldots,u_j^p\}$ and 
$\{w_j^{p+1},w_j^{p+2},\ldots,w_j^n\}$, 
the tokens can be moved to the parking structure using 
the $\calH$-type structure corresponding to the edge between $p^{th}$ 
vertex in $V_i$ and $q^{th}$ vertex in $V_j$. 
This is not possible because $X$ is an independent set.
Hence, no token in $I$ can be moved to the parking space and thus 
$I^\star$ cannot be reached from $I$.
\end{proof}

\begin{lemma}
\label{lemma:backward-reduction-leafage}
If $(G,\langle V_1,V_2,\ldots,V_k\rangle, k)$ is a \no\ instance of \textsc{MultiCol Ind-Set},
then $(G^\prime,nk)$ is a \yes\ instance of \textsc{TS-Connectivity}.
\end{lemma}
\begin{proof}
Suppose $(G,\langle V_1,V_2,\ldots,V_k\rangle, k)$ is a \no\ 
instance of \textsc{MultiCol Ind-Set}. 
We show that any independent set $I$ in $G^\prime$
of size $nk$ can be transformed to $I^\star$ via a 
sequence of valid sequence of token sliding operations. 

Note that a token on $b^{\star}$ can be moved to a parking structure.
Also, any token on any of the $\calH$-type structured vertex, or pink
or orange vertices can either be moved to $b^{\star}$ if possible or
can be moved to the blue vertices on $T_i$ 
(and this will only free up space for other tokens to move). 
Hence without loss of generality we can assume that all the tokens in $I$ are 
on $T_i$ for some $i\in[k]$ or in the parking structure. 
If there are some tokens in the parking structure, 
we will move them as close to $t_p$ as possible.
We now consider the following two mutually disjoint 
and exhaustive cases.

\emph{Case $I$:} There exists $i\in [k]$ such that $T_i$ contains 
at most $(n-1)$ tokens. 
This can happen either because some tokens are already in the parking 
structure or there exists $j\in [k]$ such that $T_j$ 
contains at least $n+1$ tokens. 
In this case, move all the tokens in $T_i$ to the parking structure 
using the pink vertices (as seen in \Cref{cl:lessthann}). 
After this, move all the vertices in each $T_j$ such that 
$j\neq i$ to the parking structure using the orange vertices.
Note that by the construction, the orange vertex intersecting both 
$T_i$ and $T_j$ can be used to move tokens in $T_j$ to $b^{\star}$
(and then to the parking structure) if and only if 
$T_i$ does not contain any token.
Thus we can reach from $I$ to $I^\star$.

\emph{Case $II$:} For every $i\in [k]$, 
there are exactly $n$ tokens in each $T_i$. 
Without loss of generality, we can assume that these 
tokens are on $\{u_i^1,u_i^2,\ldots,u_i^{p_i}\}$ and 
$\{w_i^{p_i+1},w_i^{p_i+2},\ldots,w_i^n\}$ for some $p_i\in[n]$.
Consider set formed by taking the $p_i^{th}$ vertex $v_i^{p_i}$
in the $i^{th}$ partition in $G$ for $i\in [k]$. 
This forms a $k$-sized multicolored subset of $V(G)$ 
and thus it cannot be an independent set as we 
have a \no\ instance of \textsc{MultiCol Ind-Set}.
Thus there must be an edge between two of the vertices in this set.
We will use the $\calH$-type structure corresponding to 
this edge in $G^\prime$ as seen in \Cref{cl:edge} 
to move a token in some $T_i$ 
(where the edge is incident on a vertex from the $i^{th}$ partition in $G$) 
to the parking structure. 
Now there are $n-1$ tokens in $T_i$ and we proceed similar to 
Case $I$ and thus we can reconfigure $I$ to $I^\star$ 
via a sequence of valid token sliding moves.
\end{proof}
The proof of \Cref{thm:W-hardness-leafage-TS-Conn} follows
from \Cref{lemma:forward-reduction-leafage},
\Cref{lemma:backward-reduction-leafage} and the facts
that the reduction takes polynomial time and the leafage of
the resulting chordal graph is $2k + 1$.

%% file: TS-reach-W1hard.tex

\subsection{Hardness of Token Sliding Reachability}
In this subsection, we prove \Cref{thm:W-hardness-leafage-TS-Reach}.
We present a parameter-preserving reduction
that takes an instance
$(G,\langle V_1,V_2,\ldots,V_k\rangle, k)$ of the
\textsc{MultiColored Clique} problem as an input and returns an instance
$(G^\prime, I,J)$ of the \textsc{TS-Reachability} problem.
Without loss of generality, we assume that each $V_i$
has at least one edge incident to it.
As before, we find it convenient to describe a chordal graph
$G^\prime$ as a tree model $\mathcal{T}$.
Recall that each vertex of $G^\prime$ is associated
with a specific subtree of
$\mathcal{T}$ and two vertices of $G^\prime$ share
an edge if and only if their corresponding subtrees have a non-empty intersection.
We describe the reduction 
after an informal overview.

\begin{figure}[t]
\centering
\includegraphics[scale=0.5]{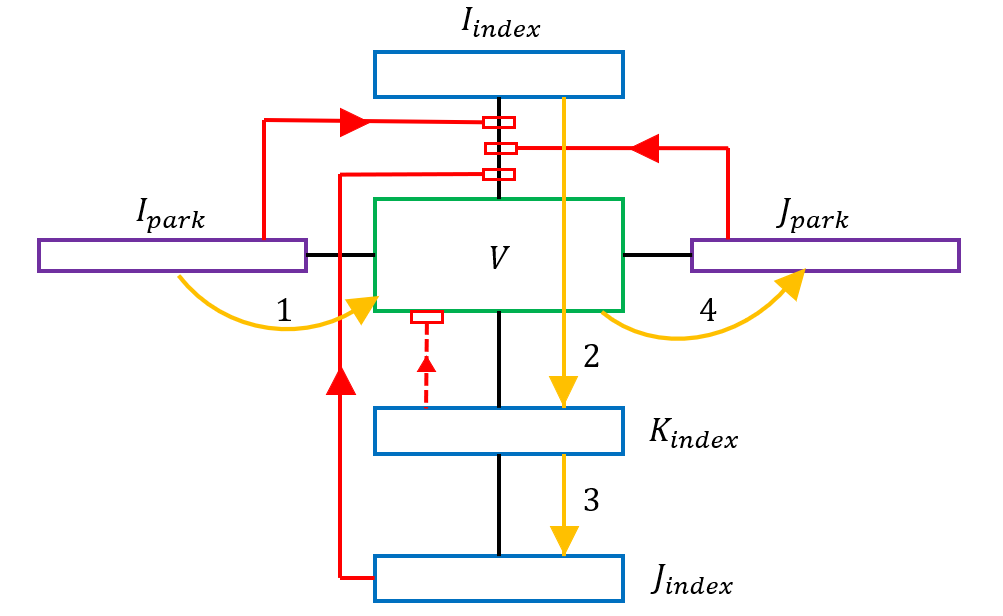}
\caption{Outline of the resulting graph $G^{\prime}$. See the informal overview of the reduction.}
\label{TSR_leafW[1]H_Red}
\vspace{-5mm}
\end{figure}

\subparagraph*{Informal Overview of the Reduction:}
\Cref{TSR_leafW[1]H_Red} shows the outline of the resulting graph $G^{\prime}$.
The green box, denoted by $V$, corresponds to the encoding of vertices in $G$.
The purple and blue boxes contain additional vertices.
Most of the vertices encoding edges in $G$,
are across purple, blue, and green boxes,
and are used to move tokens between these boxes.
The initial independent set $I$ is
$I_{\text{index}}\cup I_{\text{park}}$ and
the final independent set $J$ is
$J_{\text{index}}\cup J_{\text{park}}$.
The reduction constructs graph $G^\prime$ that satisfies
the following properties.
\begin{enumerate}
\item If there is a token in
$I_{\text{park}} \cup J_{\text{park}} \cup K_{\text{index}}$,
then one can not move a token out of $I_{\text{index}}$.
\item {All} the tokens in $I_{\text{index}}$ can be
moved to $K_{\text{index}}$ via $V$ only if
tokens in $V$ (which are moved from $I_{\text{park}}$)
corresponds to a multicolored clique in $G$.
\item Tokens in $K_{\text{index}}$ impose restrictions on movements of tokens in $V$.
\end{enumerate}
These properties imply that to move \emph{all} the tokens
from $I$ to $J$, one needs to move the tokens in the following phases:
In the first phase, one needs to move
all the tokens from $I_{\text{park}}$
to $V$ in such a way that their position corresponds to
a multicolored clique in $G$.
In the second phase, one needs to move all the tokens from
$I_{\text{index}}$ to $K_{\text{index}}$.
This move \emph{locks} the tokens at their places in $V$
and the `clique'-configuration is maintained throughout the movements
of the tokens.
As all the tokens are now out of $I_{\text{index}}$,
one can move all the tokens from
$K_{\text{index}}$ to $J_{\text{index}}$ in the third phase.
Finally, as there are no tokens in $K_{\text{index}}$ now,
one can move all the tokens in $V$ to $J_{\text{index}}$
in the fourth phase.

\begin{figure}[t]
\centering
\includegraphics[scale=0.4]{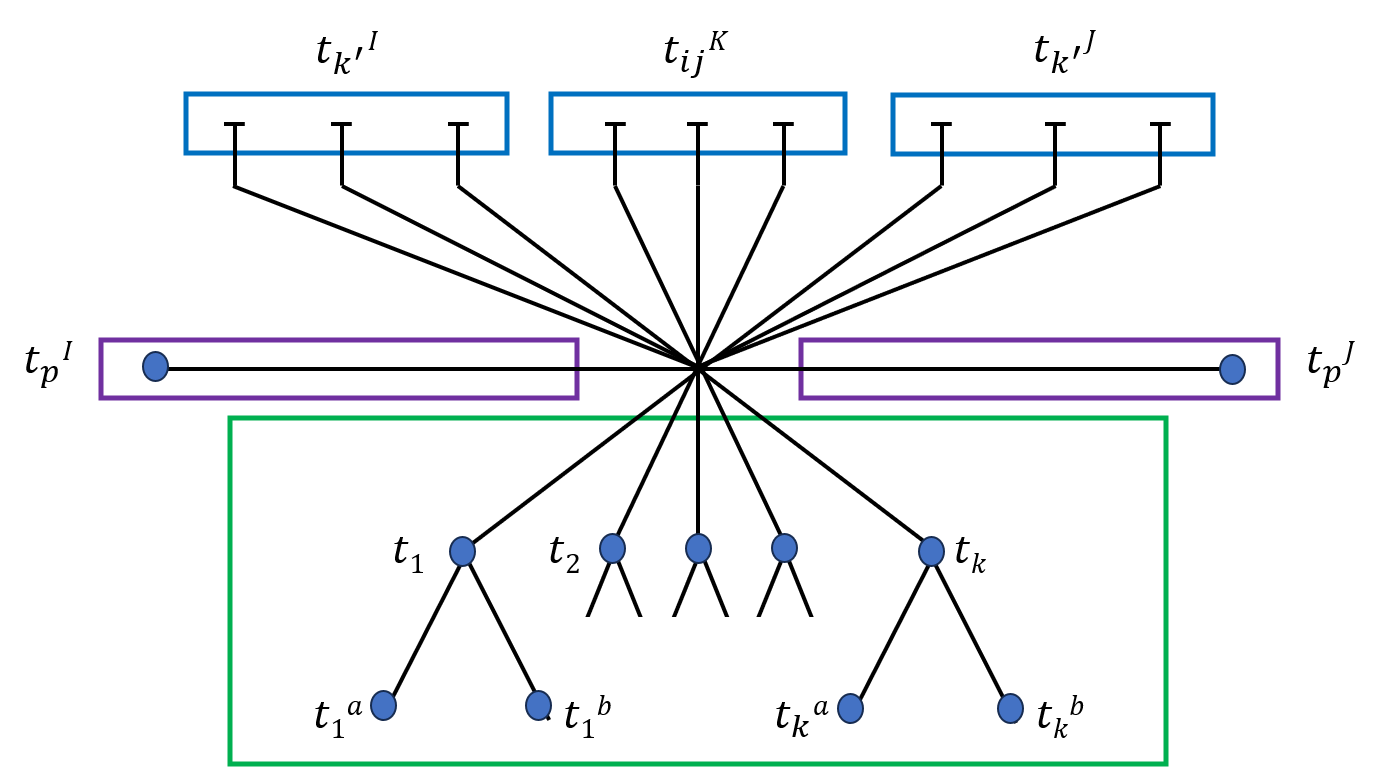}
\caption{The tree model used in the reduction to prove \Cref{thm:W-hardness-leafage-TS-Reach}. The central vertex is $t_0$.
The vertices in the green rectangles are used to encode vertices
in $V(G)$.
The purple rectangles are parking structures for the initial and final
independent set.
The vertices in the blue rectangle are index vertices, making
all of them can only be moved if placements of tokens in
the green rectangle corresponds to a clique in $V(G)$.}
\label{fig:TS-Reach-reduction-reach-tree-structure}
\vspace{-5mm}
\end{figure}

\subparagraph*{Structure of the model tree $\mathcal{T}$:}
To begin with, the model tree $\mathcal{T}$ consists of
a central vertex $t_0$, which is the root and has
$k$ children $t_1,t_2,\ldots,t_k$.
See \Cref{fig:TS-Reach-reduction-reach-tree-structure} for an illustration.
For each $i\in[k]$, $t_i$ has two children $t_i^a$ and $t_i^b$.
The reduction subdivides each edge $t_it_i^a$ and $t_it_i^b$ of the tree $\mathcal{T}$ by adding $2n-1$ new vertices on this edge.
Denote the path between $t_i^a$ and $t_i^b$ by $T_i$ for each $i \in [k]$.
We use $T_i$ to encode vertices in $V_i$.
The collection of vertices on the subtrees of the tree induced by the union of
the vertex sets $\{t_1,t_2,\ldots,t_k\}$, $\{t^a_1,t^a_2,\ldots,t^a_k\}$,
and $\{t^b_1,t^b_2,\ldots,t^b_k\}$ is denoted by $V$.
These vertices encode the vertices $V(G)$ of the input graph $G$.
Add two more children of $t_0$ and label them as $t_P^I$ and $t_P^J$.
We will use $t_P^I$ to park the tokens at the beginning and
$t_P^J$ to park the tokens at the end.

For every $k' \in [\binom{k}{2}]$, add a child of $t_0$ labelled $t^I_{k'}$
and a child of $t_0$ labelled $t^J_{k'}$.
For every $i < j \in [k]$, a child of $t_0$ labelled $t^K_{ij}$.
{We denote $I^{\calT }_{\text{index}} := \{t^I_{k'} \mid\ k' \in [\binom{k}{2}]\} $,
$J^{\calT}_{\text{index}} := \{t^J_{k'} \mid\ k' \in [\binom{k}{2}]\}$ but
$K^{\calT }_{\text{index}} := \{t^K_{ij} \mid\ i < j \in [k]\}$.}
{We highlight the two different ways of indexing these sets.}
A vertex corresponding to $t^K_{ij}$ is associated with
a collection of edges with one endpoint in $V_i$ and another endpoint
in $V_j$.
Whereas, the vertices in $I^{\calT }_{\text{index}}$ and 
$ J^{\calT}_{\text{index}}$
are a collection of $\binom{k}{2}$-many vertices each.
Subdivide each edge of the form $t^I_{k'}t_0$ and $t^J_{ij}t_0$
by adding intermediate vertices $\ell^I_{k'}$ and $\ell^J_{k'}$, respectively.
Similarly, subdivide each edge of the form $t^K_{ij}t_0$ by adding an
intermediate vertex $\ell^K_{ij}$.

\subparagraph*{Encoding the vertices of $G$:}
We add the following two types of vertices in $G^{\prime}$
to encode $n$ vertices in $V_i$ for each $i \in [k]$.
Recall that we have subdivided edge $t_it_i^a$ and edge $t_it_i^b$
of the tree $\mathcal{T}$ and added $2n-1$ new vertices
on each of these edges.
\begin{itemize}
\item Add $n$ new vertices $u_i^1,u_i^2,\ldots,u_i^n$
in $G^\prime$ corresponding to the $n$ disjoint intervals
on this edge, starting from the interval from $t_i^a$
to the first new vertex as shown in \Cref{fig:Encoding_Ver_TSR}.
Also, for every $p \in [n-1]$, add a connector vertex that connects
$u_i^p$ and $u_i^{p+1}$.
\item  Add $n$ new vertices $w_i^1,w_i^2,\ldots,w_i^n$
in $G^\prime$ corresponding to the $n$ disjoint intervals
on this edge, starting from the first new vertex to the interval
before $t_i^b$ as shown in \Cref{fig:Encoding_Ver_TSR}.
As before, for every $p \in [n-1]$, add a connector vertex
that connects $w_i^p$ and $w_i^{p+1}$.
\end{itemize}
Note that if there are $p$ tokens on the left-hand side of $T_i$
(i.e., between $t_i^a$ and $t_i$) it is safe to assume that they are on
$u_i^1,u_i^2,\ldots,u_i^p$. If not, then we can move them to
$u_i^1,u_i^2,\ldots,u_i^p$ by valid token sliding operations
without disturbing the tokens on the rest of the graph.
Similarly, if there are $p$ tokens on the right hand side of $T_i$
(i.e., between $t_i^b$ and $t_i$) it is safe to assume that
they are on $w_i^{n-p+1},w_i^{n-p+2},\ldots,w_i^n$.
If tokens are placed on $u_i^1,u_i^2,\ldots,u_i^p$ and
$w_i^{p+1},w_i^{p+2},\ldots,w_i^n$, it corresponds to
selecting the $p^{th}$ vertex in $V_i$ as a part of the clique.

\begin{figure}[t]
\centering
\includegraphics[scale=0.5]{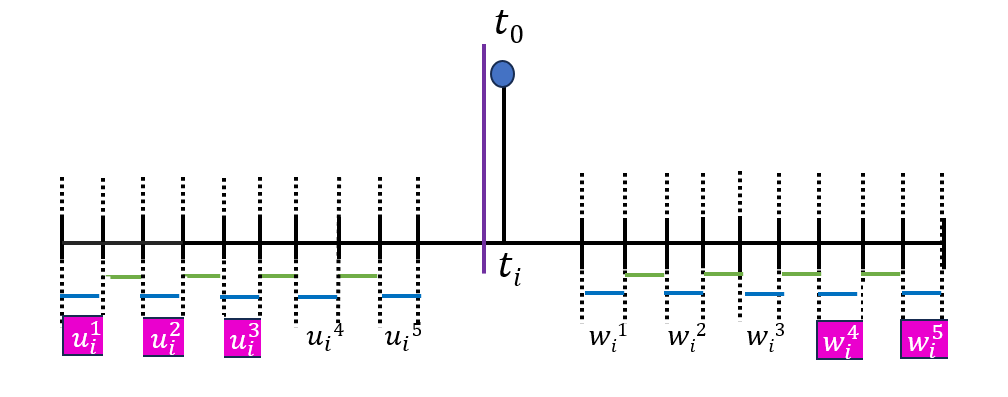}
\caption{Adding vertices corresponding to vertices in $V_i$ to $G'$ when
$n = 5$. Note that on the left side, the indexing starts from the
vertex farthest from $t_i$, whereas on the right side, the indexing
starts from the vertex closest to $t_i$. \label{fig:Encoding_Ver_TSR}
\vspace{-5mm}}
\end{figure}

\subparagraph*{Encoding the edges of $G$:}
Whenever there is an edge from the $p^{th}$ vertex in $V_i$ to the
$q^{th}$ vertex in $V_j$ where $p, q\in[n]$ and $i\neq j\in [k]$,
add a red vertex in $G^\prime$ as described below.
\begin{itemize}
\item Add a line segment starting from $u_i^{p+1}$
(including the starting point of the interval corresponding to $u_i^{p+1}$)
to the endpoint of the interval $w_i^p$
(including the endpoint of the interval corresponding to $w_i^p$).
\item Similarly, add a line segment starting from $u_j^{q+1}$
(including the starting point of the interval corresponding to $u_j^{q+1}$)
to the starting point of the interval $w_j^q$
(including the endpoint of the interval corresponding to $w_j^q$).
\item Add a horizontal line segment from $t_0$ to $t_{ij}^K$
in $K_{\text{index}}$.
\item Connect these three horizontal line segments by
a subpath from $t_i$ to $t_j$ (containing $t_0$).
\end{itemize}
See \Cref{Encoding_edge_TSR}.
We denote the vertices of $G^\prime$ which are added corresponding
to the edges of $G$ by $\calH$-type vertices as we add two structures
similar to a horizontal letter $H$ (with an extra line segment joining the
centre of this structure to $t_0$) corresponding to each edge of $G$.
We denote the $\calH$-type vertex which corresponds to the edge in $G$
corresponding to the edge from the $p^{th}$ vertex in $V_i$ and
the $q^{th}$ vertex in $V_j$ by $r_{ij}^{pq}$.
We use these vertices to move the tokens from $I_{\text{index}}$
to $K_{\text{index}}$ and from $K_{\text{index}}$ to $J_{\text{index}}$.
We make the following modifications to make the first type of movement possible.
Recall that we have subdivided each edge of the
form $t^K_{ij}t_0$ by adding an intermediate vertex $\ell^K_{ij}$.
\begin{itemize}
\item
Extend the sub-tree corresponding to each red vertex
$r^{pq}_{ij}$ till this intermediate vertex
$\ell^K_{ij}$ (or $\ell^K_{ji}$ if $i>j$).
\end{itemize}

\begin{figure}
\centering
\includegraphics[scale=0.45]{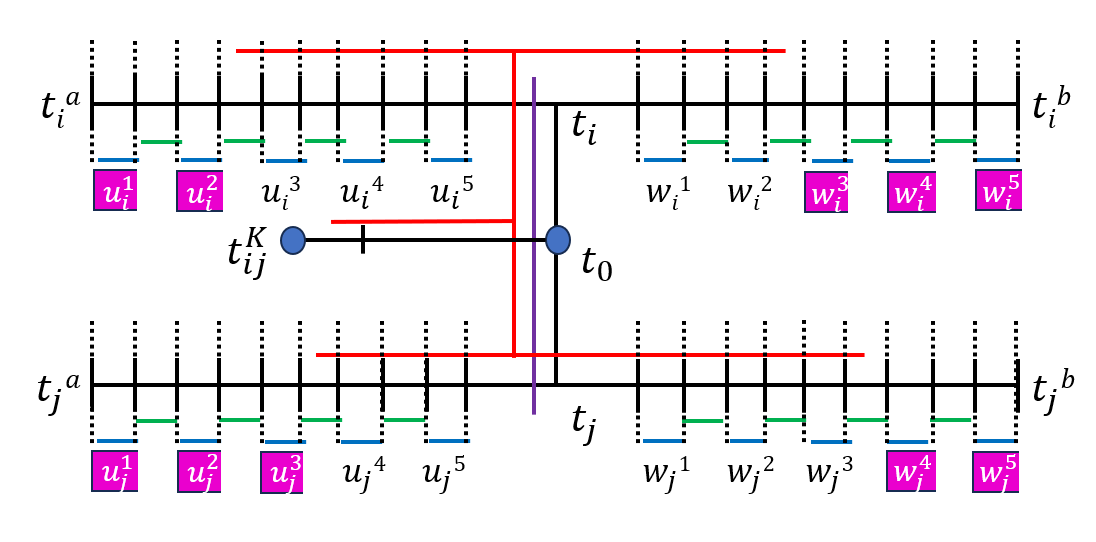}
\caption{Encoding the edge between $3^{rd}$ vertex in $V_j$ and $2^{nd}$ vertex of $V_i$.}
\label{Encoding_edge_TSR}
\vspace{-5mm}
\end{figure}

\begin{figure}[t]
\centering
\includegraphics[scale=0.45]{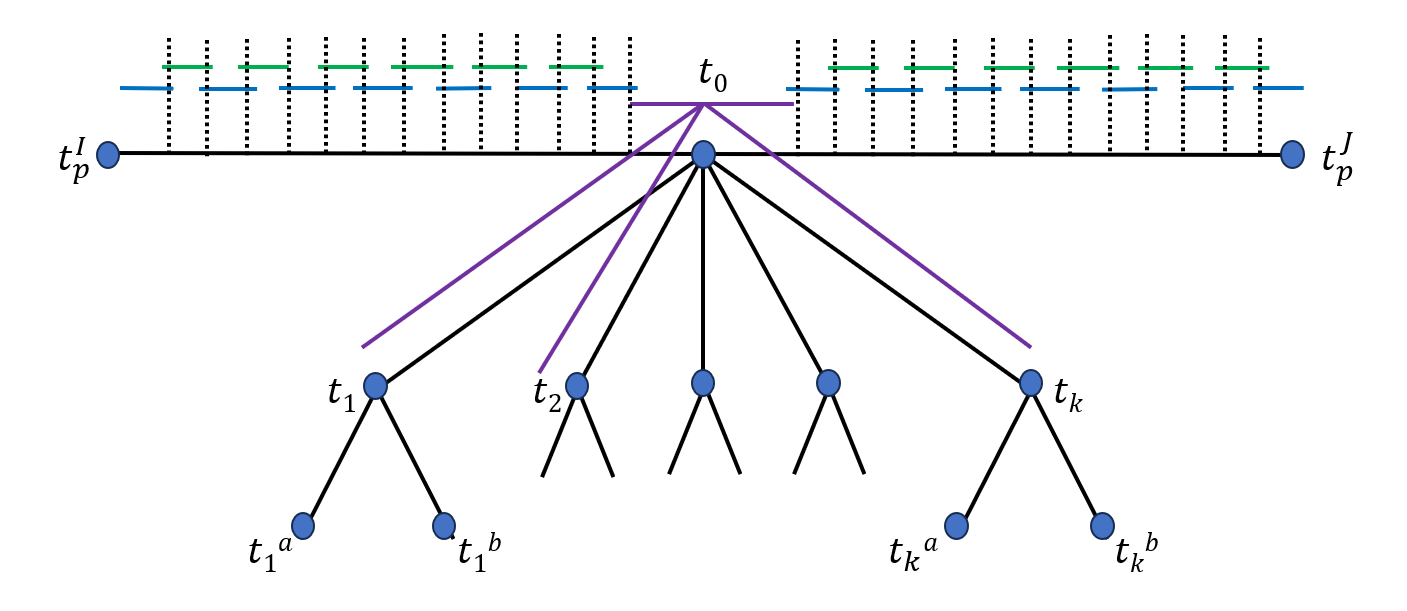}
\caption{Parking Structure; $nk$ many tokens can be parked in $J_{\text{park}}$ and $I_{\text{park}}$ each.}
\label{park_struc_TSR}
\vspace{-5mm}
\end{figure}

\subparagraph*{Auxiliary vertices:}
We start with the parking structure.
We add two parking structures to the graph
viz one by subdividing $t_0t^I_p$ and another one by
subdividing $t_0t^J_p$.
\begin{itemize}
\item We add a path in $G^\prime$ with $nk$ blue vertices
(disjoint intervals) between $t_0$ and $t_P^J$ as shown \Cref{park_struc_TSR}.
We connect these vertices using $(nk-1)$ green vertices as shown.
This allows us to park $nk$ many tokens in the
parking structure $J_\text{park}$.
Similarly, add a path with $nk$ blue vertices (disjoint intervals)
connected by $nk-1$ green vertices between $t_0$ and $t_P^I$.
\item We add a purple vertex $b^\star$ whose model is
$\{t_0,t_1,t_2,\ldots,t_k\}$.
This vertex acts as a bridge that takes tokens from the
$T_i$ part (for some $i\in [k]$) if we are able to get a token to $t_i$.
\end{itemize}
Define $I_{\text{park}}$ as the collection of the vertices
corresponding to the subpath of $\mathcal{T}$ from
$t_P^I$ to $t_0$ (including $t_P^I$ but excluding $t_0$).
Similarly, $J_{\text{park}}$ denotes the vertices corresponding
the subpath of $\mathcal{T}$ from $t_P^J$ to $t_0$
(including $t_P^J$ but excluding $t_0$).
Next, we add vertices in $I_{\text{index}}, J_{\text{index}}$ and
$K_{\text{index}}$.
\begin{itemize}
\item For every $k' \in [\binom{k}{2}]$,
add two vertices in $G^{\prime}$ whose models are
$\{t^I_{k'}\}$ and $\{\ell^I_{k'}\}$, respectively.
These are demonstrated by the dark blue and
orange intervals
in \Cref{fig:TS-Reach-auxillary-vertices}.
\item For every pair $i < j \in [k]$, add a vertex in $G^{\prime}$,
denoted $g_{ij}$, whose model is $\{t^K_{ij},\ell^K_{ij}\}$.
These are shown by the green vertices in \Cref{fig:TS-Reach-auxillary-vertices}.
\item For every $k' \in [\binom{k}{2}]$,
add a vertex, denoted by $p_{k'}$, in $G^{\prime}$
whose model is $\{t^J_{k'},\ell^J_{k'}\}$.
These are depicted by the purple vertices in
\Cref{fig:TS-Reach-auxillary-vertices}.
\end{itemize}
We now define $I_{\text{index}}:=\{\ell^I_{k^\prime}\mid k^\prime \in [\binom{k}{2}]\}$,
$J_{\text{index}}:= \{p_{k^\prime}\mid k^\prime \in [\binom{k}{2}] \}$, and
$K_{\text{index}}:=\{g_{ij}\mid i<j \in [k]\}$.
\begin{figure}
\centering
\includegraphics[scale=0.45]{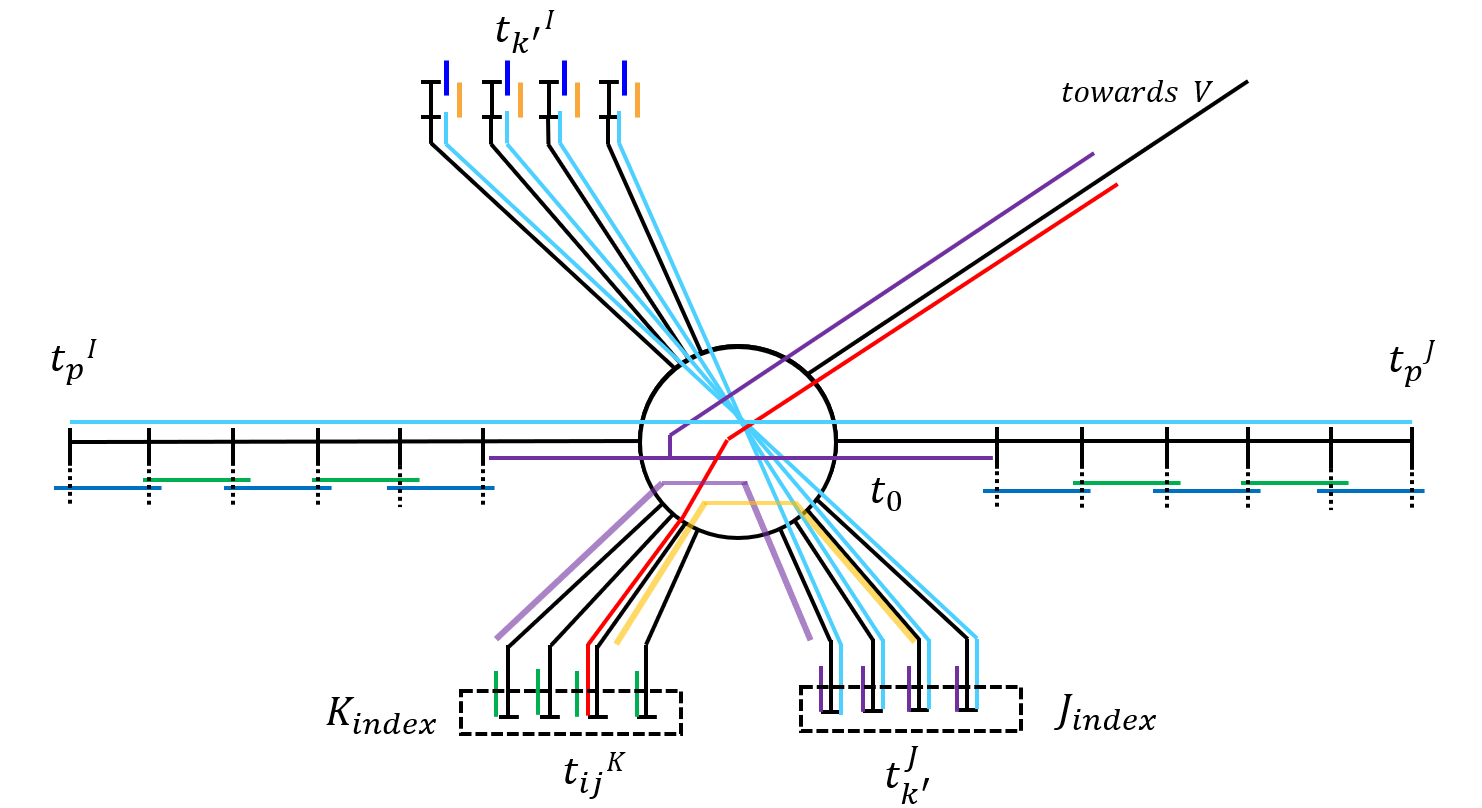}
\caption{Adding the Auxillary vertices}
\label{fig:TS-Reach-auxillary-vertices}
\vspace{-5mm}
\end{figure}
Next, we describe the construction of the
structure, which restricts the movement of tokens.
\begin{itemize}
\item Add the cyan vertex as in \Cref{fig:TS-Reach-auxillary-vertices}
which corresponds to a  star-like structure
centered at $t_0$ and covering all the vertices on the path from
$t_0$ to $t_P^I$ and on the path from $t_0$ to $t_P^J$.
This vertex also covers the paths from $t_0$ to $\ell^J_{k^\prime}$
for $k^\prime \in [\binom{k}{2}]$.
We refer to this vertex as the first \emph{choke-type vertex}
and denote it by $\calC_1$.
\end{itemize}
Finally, we describe the additional vertices in $G$ which
will facilitate the movement of the tokens from
$K_{\text{index}}$ to $J_{\text{index}}$
in case of a \yes\ instance of \textsc{Multicolored Clique}.
\begin{itemize}
\item Add a connector vertex $c^{ij}$ corresponding
to the sub-tree formed by $\ell^K_{ij},t_0$ and $\ell^J_{k^\prime}$
for each $i<j \in [k]$ and $k^\prime \in [\binom{k}{2}]$. This is shown in
 \Cref{fig:TS-Reach-auxillary-vertices}.
\end{itemize}

This completes the description of the construction of graph $G^\prime$.
The reduction sets the initial independent set $I$ is
$I_{\text{index}}\cup I_{\text{park}}$ and
the final independent set $J$ is
$J_{\text{index}}\cup J_{\text{park}}$.
Finally, it returns $(G^\prime, I, J)$ is an instance of \textsc{TS-Reachability}.
In the next two lemmas, we prove that the reduction is safe.

\begin{restatable}{lemma}{tsReachForward}
\label{lemma:forward-reduction-leafage-TS-Reach}

If $(G,\langle V_1,V_2,\ldots,V_k\rangle, k)$ is a \yes\ instance of \textsc{MultiCol Clique},
then $(G^\prime, I, J)$ is a \yes\ instance of \textsc{TS-Reachability}.
\end{restatable}
\begin{proof}
Suppose $G$ has a multicolored clique consisting of the $p_i^{th}$
vertex from each $V_i$.
We move all the tokens from $I_{\text{park}}$
to $V$ in such a way that their position corresponds to
a multicolored clique in $G$.
Formally, we move the vertices in $I_\text{index}$ to $V$ such that each
$T_i$ has tokens on $\{u_i^1,u_i^2,\ldots,u_i^{p_i}\}$ and
$\{w_i^{p_i+1},w_i^{p_i+2},\ldots,w_i^n\}$.
We say an $\calH$-type vertex is \emph{usable} (to move tokens)
if it is not adjacent to any token.
Consider an edge incident on $q^{th}$ vertex in $V_i$
and the $\calH$-type vertex added to $G^{\prime}$ to encode it.
By the construction, the interval starting from $u_i^{q+1}$ to $w_i^q$
(both inclusive) is a part of the model of $\calH$-type vertex.
This implies edges whose both endpoints are in the multicolored clique,
have no token adjacent to them even when some are moved to $V$.
Hence, the edges in the multiclique mentioned above
are always usable to move (remaining) tokens from $I$ to $V$.
And hence, we can move all $k \cdot n$ tokens from $I$ to $V$.

Using the same argument as above,
one can move all the tokens from $I_{\text{index}}$ to $K_{\text{index}}$.
Formally, consider an index $k' \in \binom{k}{2}$ and $i, j \in [k]$.
One can move a token on $\ell_{k'}^I$ in $I_{index}$
to $g_{ij}^K$ in $K_{index}$ by moving it to $b^*$
and then using the red $\calH$-type vertex corresponding to
the edge in multicolored clique across the vertex in $V_i$ and $V_j$.

In the third phase, one can move all ${k}\choose{2}$ tokens in
$K_\text{index}$ to $J_\text{index}$ by first moving a token
on $g^{ij}$ to $c^{ij}$ and then to $p^{k^\prime}$.
Note that this is possible as before starting this phase,
all the tokens are out of $I_{\text{index}}$.

Finally, as there are no tokens in $K_{\text{index}}$,
one can move all the tokens in $V$ to $J_{\text{index}}$
in the fourth phase.
If there are ${k}\choose{2}$ tokens in $J_\text{index}$ and
$nk$ tokens in $V$, all the tokens in $V$ can be brought to $J_\text{park}$.
Each token in $V$ (in some $T_i$) can be brought to the vertex
$b^*$ and then to the last unoccupied blue vertex in $J_\text{park}$.
Hence, we can move all the tokens in $K_{\text{index}}$ to $J_{\text{index}}$
and then all the tokens in $V$ to $J_{\text{park}}$.
Thus, the new independent set occupied by the tokens is $J$
which completes the proof.
\end{proof}

\begin{lemma}
\label{lemma:backward-reduction-leafage-TS-Reach}
If $(G^\prime, I, J)$ is a \yes\ instance of \textsc{TS-Reachability}
then $(G,\langle V_1,V_2,\ldots,V_k\rangle, k)$ is a \yes\ instance of \textsc{MultiCol Clique}.
\end{lemma}
\begin{proof}
We first prove some claims regarding the possible movements
of the tokens.

\begin{restatable}{claim}{tsReachclChoke}
\label{cl:choke1}
If there is a token in 
$I_{\text{park}} \cup J_{\text{park}} \cup K_{\text{index}}$,
then one can not move any token out of $I_{\text{index}}$.
\end{restatable}
\begin{claimproof}
If a token from $I_\text{index}$ must be moved outside
$I_\text{index}$, then it must be first brought to an orange vertex.
The orange vertex is adjacent to $\calC_1$.
Therefore, no vertex adjacent to $\calC_1$ can already
contain a token.
Since all the vertices in
$I_{\text{park}}\cup J_{\text{park}} \cup J_{\text{index}}$
are adjacent to $\calC_1$, if
$I_{\text{park}}\cup J_{\text{park}} \cup J_{\text{index}}$
has at least one token, then all the tokens in
$I_\text{index}$ must be on the dark blue vertices
no token can be moved to an orange vertex.
Hence no token can be moved outside $I_\text{index}$.
\end{claimproof}

\noindent Hence, to move the vertices from $I_{\text{index}}$,
we must move the vertices out of $I_{\text{park}}$.
Moreover, we can not move these tokens from $I_{\text{park}}$ 
to $J_{\text{park}}$ or to $J_{\text{index}}$ 
before all the tokens out of $I_{\text{index}}$.
We are also not be able to move any of these tokens to
$K_\text{index}$ as once a token lands on
$g_{ij}\in K_\text{index}$, the red vertex $r_{ij }$
cannot be used for any other token as it is adjacent to $g_{ij}$.
This implies that no token can be brought inside or outside
$T_i\cup T_j$.
Thus some of the tokens in $I_{\text{index}}$ will get stuck
if all of them are not moved to $K_\text{index}$ in the first phase.
Hence, one needs to move all the tokens in $I_{\text{park}}$
to $V$.
The next two claims argue about the movements
of such tokens.

\begin{restatable}{claim}{tsReachclAtmostNtokens}
\label{cl:atmostn+1tokens} 
For any $i \in [k]$,
{for an $\calH$-type vertex adjacent to some $T_i$ to remain
usable, one can move at most $n$ tokens to $T_i$.}
\end{restatable}
\begin{claimproof}
Suppose one has moved $n + 1$ tokens to  $T_i$.
Without loss of generality, we can assume that the tokens
are on $u_i^1,u_i^2,\ldots,u_i^p$ and
$w_i^n,w_i^{n-1},\ldots,w_i^p$
i.e., $p$ vertices on the left hand side of $T_i$
and $n+1-p$ vertices on the right hand side of $T_i$
where $p \in [n]$.
Consider an edge incident on $q^{th}$ vertex in $V_i$
and the $\calH$-type vertex added to $G^{\prime}$ to encode it.
By the construction, the interval starting from $u_i^{q+1}$ to $w_i^q$
(both inclusive) is a part of the model of $\calH$-type vertex.
Define $S_p = \{u_i^1,u_i^2,\ldots,u_i^p\}
\cup \{w_i^n,w_i^{n-1},\ldots,w_i^p\}$.
If $p\geq q+1$, then $S_p$ contains $u_i^{q+1}$.
Otherwise, $p\leq q$ and $S_p$ contains $w_q$.
Thus, if there are $n+1$ tokens in $T_i$, then no $\calH$-type vertex
adjacent to $T_i$ is usable.
\end{claimproof}

\noindent By the construction, at least one $\calH$-type vertex is needed
to move tokens from $t_i$ to other vertices in $T_i$.
Hence, the above claim implies that
one can move at most $n + 1$ tokens to $T_i$.

\begin{restatable}{claim}{tsReachclUsable}
\label{cl:usable}
If $T_i$ contains $n$ tokens, then $\calH$-type vertex corresponding to an edge adjacent to the $q^{th}$ vertex in $V_i$ is be usable if and only 
if it has no tokens outside $T_i$ and the tokens in $T_i$ are on 
$\{u_i^1,u_i^2,\ldots,u_i^q,w_i^n,w_i^{n-1},\ldots,w_i^{q+1}\}$.
\end{restatable}
\begin{claimproof}
Without loss of generality, we can assume that the tokens are present on
$u_i^1,u_i^2,\ldots,u_i^p$ and $w_i^n,w_i^{n-1},\ldots,w_i^{n-p+1}$
i.e., $p$ vertices on the left hand side of $T_i$ and $n-p$ vertices on the
right hand side of $T_i$ where $p\in [n]$.
If there is an edge adjacent to the $q^{th}$ vertex in $V_i$ in $G$,
then the $\calH$-type vertex corresponding to this edge
intersects $T_i$ at $\{u_i^{q+1},u_i^{q+2} \ldots,u_i^n\} \cup \
\{w_i^1,w_i^2,\ldots,w_i^q\}$.
If $p=q$, then there are $n$ tokens on
$\{u_i^1,u_i^2,\ldots,u_i^q\} \cup
\{w_i^n,w_i^{n-1},\ldots,w_i^{q+1}\}$,
then this has an empty intersection with
$\{u_i^{q+1},u_i^{q+2},\ldots,u_i^n\} \cup \{w_i^1,w_i^2,\ldots,w_i^q\}$
and thus the corresponding $\calH$-type vertex is usable
if there is no token outside $T_i$ adjacent to it.
If $p\neq q$, the intersection set $S$ of the tokens in $T_i$, i.e.,
$\{u_i^1,u_i^2,\ldots,u_i^p\} \cup
\{w_i^n,w_i^{n-1},\ldots,w_i^{n-p+1}\}$
and the vertices of $T_i$ which are adjacent to the $\calH$-type
vertex corresponding to the edge incident on the $q^{th}$ vertex
of $V_i$ given by
$\{u_i^{q+1},u_i^{q+2} \ldots,u_i^n\} \cup \{w_i^1,w_i^2,\ldots,w_i^q\}$
will be of size at least one.
If $p<q$, then $S$ will contain $w_q$ and if $p>q$,
then $S$ will contain $u_i^{q+1}$.
Therefore the $\calH$-type vertex cannot be usable.
\end{claimproof}

\noindent Now all the tokens in $I_{\text{index}}$ can be moved to
$K_{\text{index}}$ if and only if all the tokens in
$I_{\text{park}}$ are moved to $V$ such that
${k}\choose{2}$ of the red ($\calH$-type) vertices are usable
(i.e., do not have any token adjacent to them).
The red ($\calH$-type) vertices corresponding to the edges between the 
$p^{th}$ vertex in $V_i$ and the $q^{th}$ vertex in $V_j$ are usable 
for $i<j$. 
Thus we have ${k}\choose{2}$ red ($\calH$-type) vertices which are usable.

If there are ${k}\choose{2}$ red ($\calH$-type) vertices which are usable,
then each distinct pair of $i\neq j\in [k]$ must contribute one
usable red vertex.
This is because any arrangement of $n+1$ tokens on $T_i$ cannot
contribute a usable $\calH$-type vertex by \Cref{cl:atmostn+1tokens},
and an arrangement of $n$ tokens on $T_i$ can contribute at most one
usable $\calH$-type vertex.
A pair of distinct indices $i,j\in [k]$ can contribute a usable
$\calH$ type vertex $r^{pq}_{ij}$ if and only if there are tokens
on $\{u_i^1,u_i^2,\ldots,u_i^{p}\}$ and
$\{w_i^{p+1},w_i^{p+2},\ldots,w_i^n\}$ in $T_i$;
and $\{u_j^1,u_j^2,\ldots,u_j^{q}\}$ and
$\{w_j^{q+1},w_j^{q+2},\ldots,w_i^n\}$ in $T_j$.
The existence of $r^{pq}_{ij}$ implies the existence of
an edge between the $p^{th}$ vertex of $V_i$ and
the $q^{th}$ vertex of $V_j$ in $G$.
Thus the existence of ${k}\choose{2}$ red ($\calH$-type)
vertices that are usable imply the existence of a 
multicolored clique in $G$.
\end{proof}

\Cref{lemma:forward-reduction-leafage-TS-Reach} and
\Cref{lemma:backward-reduction-leafage-TS-Reach}
implies that the reduction is safe.
This, along with the fact that the reduction works in polynomial time
and the leafage $\ell$ of the chordal graph $G^\prime$
is $3\cdot {{k}\choose{2}}+2k+2$
implies the 
proof of \Cref{thm:W-hardness-leafage-TS-Reach}.

%% file: conclusion.tex
\section{Conclusion}
\label{sec:conclusion}
In this article, we studied the \textsc{Token Sliding Connectivity} 
and \textsc{Token Sliding Reachability} problems when the input 
is restricted to a chordal graph. 
We prove that both these problems are \para-\NP-hard 
when parameterized by the maximum clique-degree of the chordal graph. 
This answers the open question posed by 
Bonamy and Bousquet~\cite{DBLP:conf/wg/BonamyB17}. 
We then consider the complexity of the problems when 
parameterized by the larger parameter {`leafage'} and 
prove that \textsc{Token Sliding Connectivity} and 
\textsc{Token Sliding Reachability} are \co-\W[1]-hard and \W[1]-hard, 
respectively. 
We conjecture that both these problems admit an \XP-algorithm, i.e., an 
algorithm with running time $n^{f(\ell)}$, when parameterized by the leafage. 
It would also be interesting to {investigate} whether well-partitioned chordal 
graphs, introduced in \cite{DBLP:journals/dm/AhnJKL22}, can be used to 
narrow down the complexity gaps for these problems. 
Formally, do these problems admit \FPT\ algorithms when 
parameterized by the leafage when the input graph is restricted 
to a well-partitioned chordal graph?